\documentclass[lettersize,journal]{IEEEtran}
\usepackage{array}
\usepackage{cite}
\usepackage{amsmath,amssymb,amsfonts}
\usepackage{mathrsfs}
\usepackage{algorithmic}
\usepackage{subcaption}
\usepackage{graphicx}
\usepackage{textcomp}
\usepackage{xcolor}
\definecolor{blue}{RGB}{0, 0, 255}
\definecolor{red}{RGB}{255, 0, 0}
\usepackage{amsfonts}
\usepackage[normalem]{ulem} 
\usepackage{cancel}
\usepackage{adjustbox}

\usepackage{epstopdf}
\usepackage{textcomp}
\usepackage{stfloats}
\usepackage{multirow}
\usepackage{algorithm}
\usepackage{setspace}
\usepackage{epsfig}
\usepackage{bm}
\usepackage{ntheorem} 
\newtheorem{definition}{Definition} 
\newtheorem{theorem}{Theorem}
\newtheorem{Corollary}{Corollary} 

\newtheorem{lemma}{Lemma}  

\newtheorem{assumption}{Assumption}
\newenvironment{proof}{{\indent \indent \it Proof:}}{\hfill $\square$}
\def\BibTeX{{\rm B\kern-.05em{\sc i\kern-.025em b}\kern-.08em

    T\kern-.1667em\lower.7ex\hbox{E}\kern-.125emX}}
\begin{document}

\title{RRAM Circuit-enabled Nonlinear Precoding and Bit Precision Analysis\\
}
\author{\IEEEauthorblockN{Yuhao~Zhang, Haifan~Yin,~\IEEEmembership{Senior Member,~IEEE,} Tao~Wang, Jindiao~Huang, and Kewei Zhu}
\thanks{Part of this work has been presented in the 2025 IEEE 101st Vehicular Technology Conference (VTC2025-Spring)\cite{zyh_VTC}.}
\thanks{The authors are with the School of Electronic Information and Communications, Huazhong University of Science and Technology, 430074 Wuhan, China (e-mail: yhz@hust.edu.cn, yin@hust.edu.cn, twt@hust.edu.cn, jindiaohuang@hust.edu.cn, zhu@hust.edu.cn).}% <-this % stops a space
\thanks{The corresponding author is Haifan Yin.}
\thanks{This work was supported by Mobile Information Network-National Science and Technology Major Project under Grant 2026ZD1306400, the Major Program (JD) of Hubei Province under Grant 2025BEA001, and the Fundamental Research Funds for the Central Universities.}
}
\maketitle

\begin{abstract}
The rising number of users and antennas imposes exponentially growing computational loads on future communication systems. Yet conventional processors are facing a bottleneck for their nature of memory-computing separation. In-memory computing (IMC) emerges as a promising solution leveraging its intrinsic high parallelism. This work proposes an IMC architecture that employs resistive random access memory (RRAM) to reduce the computational complexity of the nonlinear Tomlinson-Harashima precoding (THP) to a linear scale. We present a computation-constraint principle for designing RRAM circuits to perform nonlinear matrix operations and construct an LQ decomposition RRAM circuit. Since the conductance of memristor is generally quantized, we perform the bit precision analysis and derive the lower bound of the Signal-to-Interference-plus-Noise Ratio (SINR) and achievable rate. Our analysis indicates that at a high Signal-to-Noise Ratio (SNR) or with a large number of antennas, each 1-bit precision increase yields 6 dB SINR gain and linear rate growth. For practical implementation, we derive the optimal bit precision to sustain SINR performance under varying system configurations. Simulation demonstrates the feasibility and accuracy of the RRAM-based circuit and our theoretical analysis. Our work proves that RRAM-based IMC holds significant potential for high-complexity nonlinear precoding, addressing the escalating computational demands for future communications. 
\end{abstract}


\section{Introduction}
\IEEEPARstart{M}{assive} multi-input multi-output (MIMO) technology effectively exploits the spatial resources of wireless channels and achieves substantial gains in both spectral and energy efficiency\cite{Marzetta10TCom}. However, overlapping signals from co-channel users introduce considerable multiuser interference (MUI) and degrade system performance.

Precoding algorithms, by pre-processing transmitted signals, effectively suppress MUI and emerge as a critical technique at base stations. Precoding algorithms are broadly categorized into two classes: linear precoding and nonlinear precoding. Linear precoding algorithms weight the transmitted signals using linear transformation matrices to suppress MUI, like matched-filter (MF) and zero-forcing (ZF) algorithms\cite{Precodingover3}. Nonlinear precoding introduces additional nonlinear operations to optimize signal transmission, such as dirty paper coding (DPC) and Tomlinson-Harashima precoding (THP)\cite{DPC, THP}. In general, nonlinear precoding may provide better performance than linear precoding, though it involves higher computational complexity\cite{Precodingover2}. However, with the rising number of antennas and users, the channel matrix experiences a substantial growth in dimensionality. This expansion leads to unaffordable computational complexity. When the computation time surpasses the channel coherence time, the outdated channel state information (CSI) can cause considerable MUI, which, in turn, degrades the performance of massive MIMO systems. Communication systems often employ simpler yet suboptimal precoding techniques at base stations to ensure that the precoding algorithm is deployable in practice.

Over the past few decades, with the development of novel optoelectronic materials\cite{memristor}, a revolutionary computing architecture, in-memory computing (IMC) has been proposed, which may offer a promising solution to implement high-complexity algorithms. In contrast to the conventional von Neumann architecture, where computation and memory units are physically separated, IMC enables data processing directly at storage locations, thereby eliminating redundant data migration and reducing latency and power consumption \cite{ielmini2018memory}. Based on the underlying computing principle, IMC architectures can be broadly classified into digital IMC and analog IMC. Digital IMC performs fully digital operations\cite{ReDCIM23JSSC}, whereas analog IMC exploits physical-domain properties to perform computation directly in the analog domain. Analog IMC can be implemented using different memory technologies, among which SRAM-based analog IMC and RRAM-based analog IMC are two mainstream approaches\cite{CIMoverview}. Existing studies have verified that both approaches can achieve $\mathcal{O}(1)$ complexity for matrix operations, such as matrix-vector multiplication (MVM) and matrix inversion\cite{MVM,sun2019solving, SRAMonestep,mannocciFullyIntegratedAnalogue2026}. SRAM-based analog IMC features high computational accuracy and compute density in reported works\cite{SRAMACIM21JSSC, SRAMCIMOver}. RRAM-based analog IMC is characterized by non-volatility to eliminate the energy from data retention. RRAM devices also provide potential high integration density at the bit-cell-array level\cite{ielmini2018memory}.

Recently, there has been growing interest in deploying RRAM-based analog IMC in communication systems. In \cite{wang2023parallel}, the authors construct an orthogonal frequency division multiplexing transceiver by memristor arrays and achieve a bit error rate of $0/480$ in transmission.
RRAM-based ZF precoding and detection are proposed in \cite{ZF_RRAM,zuo2025precise} and present a significant advantage in computational complexity and energy efficiency. 
The authors of \cite{RZF_RRAM} implement RRAM circuits to implement regularized zero-forcing (RZF) detection and beamforming.
In \cite{FFT}, RRAM devices are utilized to accelerate baseband processing and a mathematical model is established for evaluating the relationship between the latency of programming and antenna configurations. By integrating RRAM circuits and analog multipliers, the authors of \cite{ML} achieve computational acceleration for maximum-likelihood (ML) detection. Based on gradient descent, the authors of \cite{MMSE_RRAM} convert the classic minimum mean square error (MMSE) detection algorithm into a linearized iterative algorithm that is readily implementable by RRAM arrays. For large-scale RRAM arrays used for MIMO signal processing, the authors of \cite{THU22RRAM} propose a model to address the IR-drop effects induced by non-negligible parasitic resistance.

However, the existing literature mainly focuses on leveraging the linear combination of $\mathcal{O}(1)$-complexity MVM and matrix inversion operations of RRAM circuits to accelerate linear precoding or detection algorithms. The implementation of more complex matrix operations, such as LQ decomposition, remains challenging with RRAM devices. This limitation restricts the deployment of advanced communication algorithms, particularly nonlinear precoding algorithms. To date, few studies have explored architectures that utilize RRAM to accelerate nonlinear precoding. 

Furthermore, how finite bit precision in RRAM affects communication performance remains insufficiently characterized in the existing literature. The resistive switching behavior of RRAM relies on the migration of oxygen vacancies and the dynamic formation/rupture of conductive filaments\cite{li2015conductance}. These microscopic processes are inherently random. It is difficult to precisely program the device to arbitrary conductance levels. Therefore, practical RRAM-based IMC typically adopts finite conductance states to improve robustness against device non-idealities\cite{li2018analogue,liu2026memristor}. While prior investigations in \cite{RZF_RRAM,ML} have explored the impact of bit precision through simulations, a mathematical framework has yet to be established.

In this paper, we propose an IMC architecture that employs RRAM circuits to implement the nonlinear THP algorithm and reduce its computational complexity. Firstly, we present a computation-constraint principle to perform nonlinear matrix operations by RRAM circuits. Following the principle, an RRAM-based LQ decomposition acceleration circuit with the complexity of $\mathcal{O}(1)$ is presented. Based on the RRAM circuits, we propose the RRAM-based THP algorithm with linear complexity. Since the memristor generally has limited bit precision, we consider the quantization error and derive the lower bound of Signal-to-Interference-plus-Noise Ratio (SINR) and achievable rate performance for the limited-precision RRAM-based THP algorithm. Finally, we present simulation results to validate the performance of the RRAM-based THP algorithm and the accuracy of our theoretical analysis.

The main contributions of our work are as follows:
\begin{itemize}
\item{We propose an RRAM-based IMC computing architecture for nonlinear THP algorithms. We conclude a computation-constraint principle to perform nonlinear matrix operations by RRAM circuits and design an RRAM-based LQ decomposition acceleration circuit with the complexity of $\mathcal{O}(1)$. The proposed circuit exhibits a significant throughput advantage over graphics processing units (GPUs) and field-programmable gate arrays (FPGAs). The complexity of the RRAM-based THP algorithm is reduced from $\mathcal{O}(K^3)+\mathcal{O}(K^2N)$ to $\mathcal{O}(K)$. }
\item{We prove that quantization errors caused by the finite bit precision of memristors are approximately uniformly distributed. Based on it, we derive closed‑form lower bounds for the SINR and achievable rate of limited‑precision RRAM‑based THP. This derivation reveals the existence of a ceiling effect for both the number of antennas and Signal-to-Noise Ratio (SNR). We demonstrate that under high‑SNR or large‑antenna regimes, the SINR performance exhibits a characteristic $6$ dB improvement per additional bit precision and the achievable rate shows an approximately linear growth trend with respect to the bit precision. }
\item{We define quantization efficiency coefficient as a universal metric to evaluate bit precision under varying SNR conditions and antenna configurations. We derive the optimal bit precision to sustain system performance across various requirements of the quantization efficiency coefficient. This may offer a quantitative guideline for establishing bit precision values in practical applications.}
\end{itemize}

The rest of this paper is organized as follows. In Sec. \uppercase\expandafter{\romannumeral2}, the system model of the THP algorithm is introduced. In Sec. \uppercase\expandafter{\romannumeral3}, we construct an RRAM-based LQ decomposition acceleration circuit and provide the system architecture of the RRAM-based THP algorithm. In Sec. \uppercase\expandafter{\romannumeral4}, we analyze the impact of the bit precision of memristors on system performance and derive the closed-form expression of the lower bound of SINR and achievable rate in the RRAM-based THP algorithm. The simulation results are provided in Sec. \uppercase\expandafter{\romannumeral5} and the final conclusions are described in Sec. \uppercase\expandafter{\romannumeral6}.

\textit{Notation:} Matrices and vectors are denoted by boldface uppercase letters and boldface lowercase letters, respectively. For a matrix $\mathbf{X}$, its conjugate, transpose, conjugate transpose and determinant are denoted by $(\mathbf{X})^{*}$, $(\mathbf{X})^{T}$, $(\mathbf{X})^{H}$ and $\operatorname{det}(\mathbf{X})$, respectively. $\text{diag}\{x_1,\dots,x_N\}$ is a diagonal matrix with $x_1,\dots,x_N$ at the main diagonal. The $ \ell_2 $ norm of a vector is denoted by $\|\cdot\|_2$. $\mathbf{I}_K$ denotes the $K \times K$ identity matrix. $\mathbb{C}$ and $\mathbb{R}$ denote the sets of complex and real numbers, respectively. $\mathcal{CN}(\mu,\sigma^2)$ stands for the complex univariate Gaussian distribution with the mean $\mu$ and variance $\sigma^2$. $\left \lfloor\cdot\right \rfloor$ and $\left \lceil\cdot\right \rceil$ denote the floor and ceiling functions, respectively. $\mathbb{E}\{\cdot\}$ represents the expectation. 

\section{System Model}
In this paper, we consider the multi-user downlink transmission, with $N$ antennas at the base station and $K$ single-antenna users such that $N \ge K$. Let $\mathbf{H}\in\mathbb{C}^{K \times N}$ be the channel matrix between the base station and users with independent and identically distributed (i.i.d.) complex Gaussian elements $\mathcal{CN}(0,\sigma_h^2)$. 

Fig. \ref{THPmodel} is the system model of the THP algorithm. The feedforward precoding filter $\mathbf{F} \in \mathbb{C}^{N \times K}$ is used to perform a spatial channel pre-equalization\cite{THP}. The feedback filter $\mathbf{B} \in \mathbb{C}^{K \times K}$, constrained to be a unit-diagonal lower triangular matrix, successively eliminates interference caused by previous users. The scaling filter $\mathbf{G} \in \mathbb{C}^{K \times K}$ is a diagonal matrix to compensate for the channel gain. To obtain the filters 
\begin{figure}[t]
\centering
\includegraphics[width=0.45\textwidth]{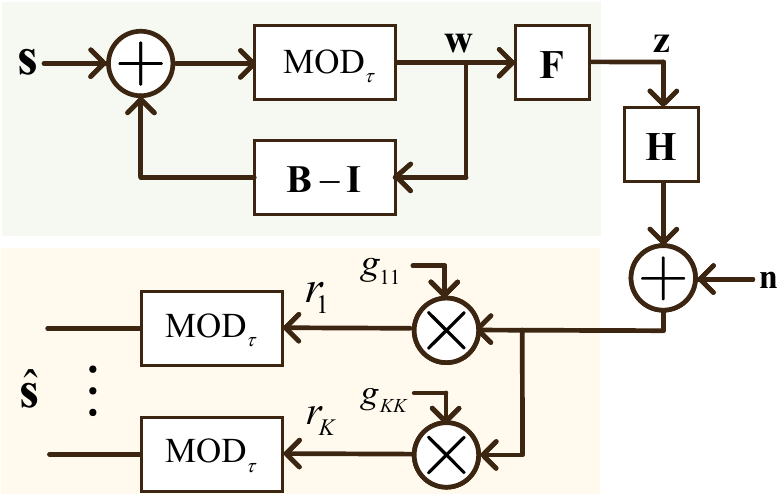}
\caption{System model of THP algorithm}
\label{THPmodel}
\end{figure}in the THP algorithm, LQ decomposition is performed on the channel matrix.
\begin{equation}
\label{LQ}
    \mathbf{H}=\mathbf{L}\mathbf{Q},
\end{equation}
where $\mathbf{L} \in \mathbb{C}^{K \times K}$ and $\mathbf{Q} \in \mathbb{C}^{K \times N}$. Three filters are obtained as follows:
\begin{equation}\label{cal_F}
    \mathbf{F}=\mathbf{Q}^H,
\end{equation}
\begin{equation}
\mathbf{G}=\text{diag}\{l_{11}^{-1},\cdots,l_{KK}^{-1}\},
\end{equation}
\begin{equation}\label{cal_B}
\mathbf{B}=\mathbf{G}\mathbf{L},
\end{equation}
where $l_{ij}$ is the element of $\mathbf{L}$ in row $i$ and column $j$.

Let vector $\mathbf{s}=[s_1,...,s_K]^T \in \mathbb{C}^{K \times 1}$ represent the modulated signal vector for all users. The channel signals $\mathbf{w}=[w_1,...,w_K]^T\in \mathbb{C}^{K \times 1}$ are successively generated as
\begin{equation}
\label{mod}
w_i=\mathrm{MOD}_\tau (s_i-\sum_{j=1}^{i-1}b_{ij}w_j),\quad i=1,...K,
\end{equation}
where $b_{ij}$ is the element of $\mathbf{B}$ in row $i$ and column $j$. $\mathrm{MOD_\tau(\cdot)}$ is a modulo operation that limits the amplitude of $w_i$ to the boundary of the modulation alphabet\cite{THPMOD}, expressed as
\begin{equation}
\mathrm{MOD}_\tau (w)=w- \left \lfloor \frac{\operatorname{Re}(w)}{\tau}+\frac{1}{2} \right \rfloor\tau-j \left \lfloor \frac{\operatorname{Im}(w)}{\tau}+\frac{1}{2} \right \rfloor\tau ,
\end{equation}
where $\tau$ is a constant depending on the modulation. With the modulo operation, $\mathbf{w}$ is represented as
\begin{equation}
\mathbf{w}=\mathbf{B}^{-1}\left(\mathbf{s}+\mathbf{d}\right)=\mathbf{B}^{-1}\mathbf{v},
\end{equation}
where $\mathbf{d}$ is a perturbation vector and $\mathbf{v}$ is an equivalent signal.

By left-multiplying feedforward matrix $\mathbf{F}$, $\mathbf{w}$ undergoes a spatial channel pre-equalization to obtain the transmit signal $\mathbf{z}$, as 
\begin{equation}\label{pre-equalization}
\mathbf{z}=\mathbf{F}\mathbf{w}.
\end{equation}

At the receiver side, the received signal $\mathbf{r}$ is compensated by the scaling matrix $\mathbf{G}$, which takes the form
\begin{equation}
\label{receive}
\mathbf{r}=\mathbf{G}(\mathbf{H}\mathbf{F}\mathbf{w}+\mathbf{n})=\mathbf{v}+\mathbf{G}\mathbf{n},
\end{equation}
where $\mathbf{n}\in \mathbb{C}^{K \times 1}$ is an additive white Gaussian noise vector.

\section{RRAM-based circuit for THP algorithm}
In this section, we propose an RRAM-based system architecture to reduce the computational complexity of the THP algorithm.

\subsection{Analysis of the computational complexity}
Although the THP algorithm effectively mitigates MUI in downlink channels, its computational complexity poses significant implementation challenges for the base station. As quantified in \cite{THPtime}, the total complexity is expressed as:
\begin{equation}
C_{\mathrm{THP}}=16K^3/3+8K^2N+2K^2+8KN+6K-2N-8.
\end{equation}

For the base station, this complexity arises from four primary components in the THP processing:
\begin{enumerate}
    \item LQ decomposition on the channel matrix (Eq. (\ref{LQ})).
    \item Filter computations (Eq. (\ref{cal_F})-(\ref{cal_B})).
    \item Channel signal generation (Eq. (\ref{mod})).
    \item Channel pre-equalization (Eq. (\ref{pre-equalization})).
\end{enumerate}

Let $\mathbf{C}=\mathbf{B}-\mathbf{I}_K=[\mathbf{c_1}^T,...,\mathbf{c_K}^T]^T$, $\mathbf{w}^{\text{fru}}_i=[w_1,...,w_i,0,...,0]^T$. From Eq. (\ref{mod}), the third step has the following equivalent form 
\begin{equation}
\label{MVMform_mod}
w_i=\mathrm{MOD}_\tau(s_i-\mathbf{c_i}\mathbf{w}^{\text{fru}}_{i-1}).
\end{equation}

From Eq. (\ref{cal_B}), (\ref{pre-equalization}) and (\ref{MVMform_mod}), the latter three steps of the THP algorithm can be regarded as a combination of high-complexity MVM operations and other simple computational tasks. Therefore, the computational complexity of the THP algorithm mainly comes from LQ decomposition and MVM operation. It has been proved that the RRAM array has an impressive effect to reduce the computational complexity of MVM operation to $\mathcal{O}(1)$\cite{MVM}. In this paper, we mainly focus on the RRAM-based acceleration circuit for LQ decomposition.

\subsection{Realification method}

Before presenting the proposed circuit design, we first introduce a standard technique for implementing complex-valued operations. An RRAM array performs analog computation by mapping matrix elements as nonnegative conductance values. Consequently, complex-valued matrix operations cannot be mapped directly onto a single RRAM array. We adopt the standard real-valued representation of complex linear operations in \cite{ZF_RRAM}, which converts a complex-valued system into an equivalent real-valued system. For an arbitrary complex matrix $\boldsymbol{\Psi}\in\mathbb{C}^{m\times n}$ and complex vector $\boldsymbol{\omega}\in\mathbb{C}^{n\times1}$, define their real-valued representations as
\begin{align}
    \mathcal{R}_{\mathrm m}(\boldsymbol{\Psi})
    \triangleq
    \begin{bmatrix}
        \operatorname{Re}(\boldsymbol{\Psi}) &
        -\operatorname{Im}(\boldsymbol{\Psi})\\
        \operatorname{Im}(\boldsymbol{\Psi}) &
        \operatorname{Re}(\boldsymbol{\Psi})
    \end{bmatrix},
    \label{eq:Rm_definition}
\end{align}
\begin{align}
    \mathcal{R}_{\mathrm v}(\boldsymbol{\omega})
    \triangleq
    \begin{bmatrix}
        \operatorname{Re}(\boldsymbol{\omega})\\
        \operatorname{Im}(\boldsymbol{\omega})
    \end{bmatrix}.
    \label{eq:Rv_definition}
\end{align}
Conversely, for any
$\boldsymbol{\chi}
=[\boldsymbol{\chi}_{\mathrm R}^{T},
\boldsymbol{\chi}_{\mathrm I}^{T}]^{T}$, where
$\boldsymbol{\chi}_{\mathrm R},
\boldsymbol{\chi}_{\mathrm I}\in\mathbb{R}^{n\times1}$, the inverse mapping of
$\mathcal{R}_{\mathrm v}$ is defined as
\begin{align}
    \mathcal{R}_{\mathrm v}^{-1}(\boldsymbol{\chi})
    \triangleq
    \boldsymbol{\chi}_{\mathrm R}
    +\mathrm{j}\boldsymbol{\chi}_{\mathrm I}.
    \label{eq:Rv_inverse}
\end{align}

\begin{figure*}[h]
\centering
\includegraphics[width=0.93\textwidth]{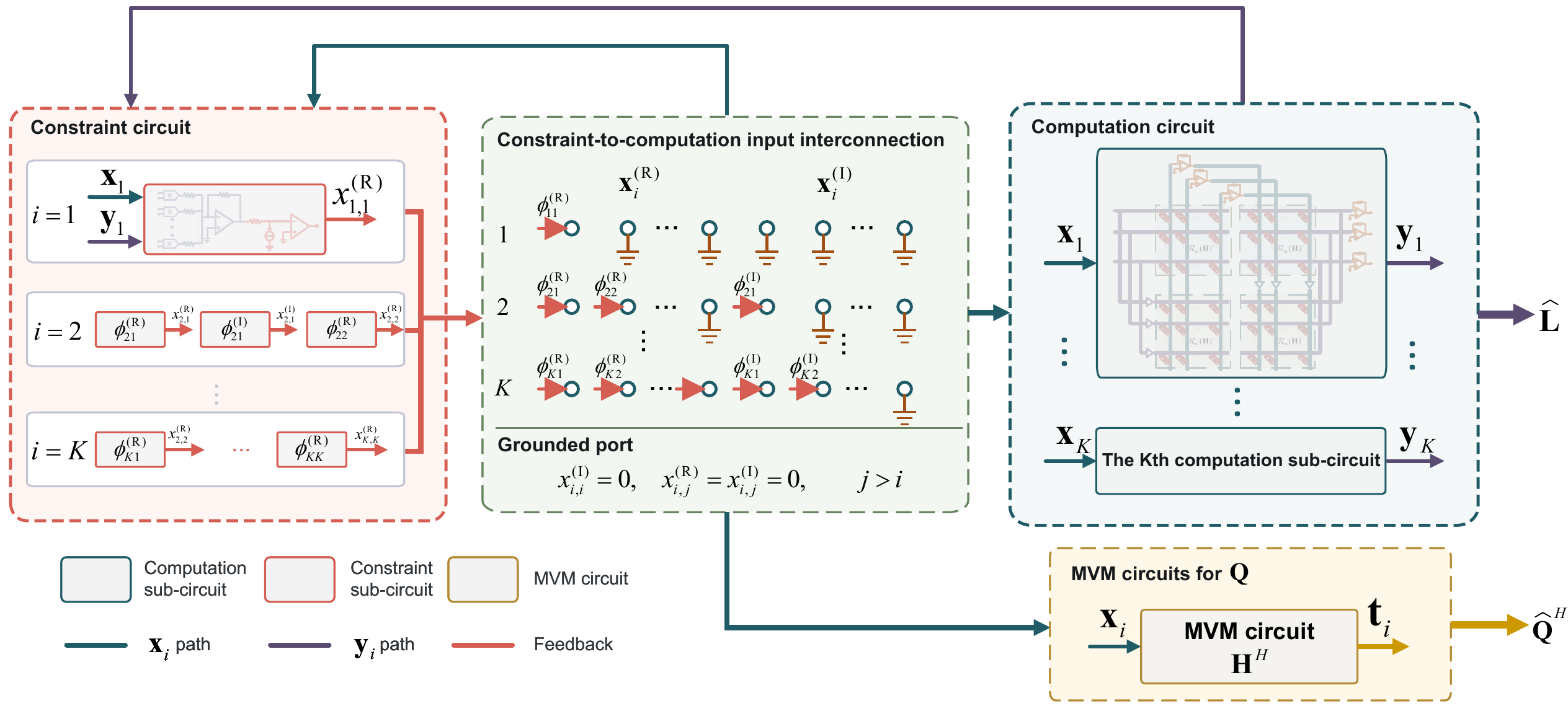}
\caption{The overall architecture of the proposed RRAM-based LQ decomposition circuit.}
\label{overallcircuit}
\end{figure*}

They satisfy the following expression
\begin{align}
    \boldsymbol{\Psi}\boldsymbol{\omega}=\mathcal{R}_{\mathrm v}^{-1}\left(\mathcal{R}_{\mathrm m}(\boldsymbol{\Psi})\mathcal{R}_{\mathrm v}(\boldsymbol{\omega})\right).\label{eq:realification_identity}
\end{align}

Eq. (\ref{eq:realification_identity}) allows a complex-valued MVM to be implemented as an equivalent real-valued MVM. $\mathcal{R}_{\mathrm m}(\boldsymbol{\Psi})$ may still contain negative entries. It is further decomposed as
\begin{align}
    \mathcal{R}_{\mathrm m}(\boldsymbol{\Psi})=\mathcal{R}_{\mathrm m}(\boldsymbol{\Psi})^{+}-\mathcal{R}_{\mathrm m}(\boldsymbol{\Psi})^{-},
    \label{eq:complex_matrix}
\end{align}
where
\begin{align}
\left[\mathcal{R}_{\mathrm m}(\boldsymbol{\Psi})^{+}\right]_{ij}&=
\max\left([\mathcal{R}_{\mathrm m}(\boldsymbol{\Psi})]_{ij},0
\right),\\
\left[\mathcal{R}_{\mathrm m}(\boldsymbol{\Psi})^{-}\right]_{ij}&=-\min\left(
[\mathcal{R}_{\mathrm m}(\boldsymbol{\Psi})]_{ij},0
\right).
\label{eq:positive_negative_parts}
\end{align}
Both $\mathcal{R}_{\mathrm m}(\boldsymbol{\Psi})^{+}$ and
$\mathcal{R}_{\mathrm m}(\boldsymbol{\Psi})^{-}$ are nonnegative and can therefore be mapped to RRAM arrays. The subtraction in Eq. (\ref{eq:complex_matrix}) is readily implementable with inverters. This real-valued formulation has also been adopted in prior RRAM-based accelerators to execute complex-valued operations \cite{ZF_RRAM,zuo2025precise,RZF_RRAM,FFT,ML} and is utilized below to implement complex matrix LQ decomposition.

\subsection{RRAM-based LQ decomposition circuit}
We propose a computation-constraint principle for an RRAM-based acceleration circuit, which may provide systematic guidance for designing nonlinear operation circuits. We deconstruct the RRAM circuit into two parts: the computation circuit and the constraint circuit. The computation circuit, typically composed of RRAM devices, is designed to store the data matrix and output target voltages efficiently in a parallel manner. Meanwhile, the constraint circuit, which is realized through additional circuit components like operational amplifiers (OAs), serves to establish nonlinear constraints for the computational quantities in the computation circuit and complete negative feedback closure. This mechanism guarantees that the output voltage converges to the desired target value.

Following this principle, we construct the RRAM-based LQ decomposition circuit. Fig. \ref{overallcircuit} presents the overall circuit. The core architecture consists of the computation circuit and the constraint circuit. The constraint circuit takes both the input and output voltage vectors of the computation circuit as its inputs. The OA outputs in the constraint circuit provide the computation-input voltage vectors through feedback paths. The voltage outputs of the computation circuit form the matrix $\hat{\mathbf{L}}$, while selected internal nodes are connected to additional RRAM-based MVM circuits to form another output matrix $\hat{\mathbf{Q}}$. These two circuit outputs will be shown to constitute a valid LQ decomposition of $\mathbf{H}$.

\begin{figure}[!htp]
\centering
\includegraphics[width=0.48\textwidth]{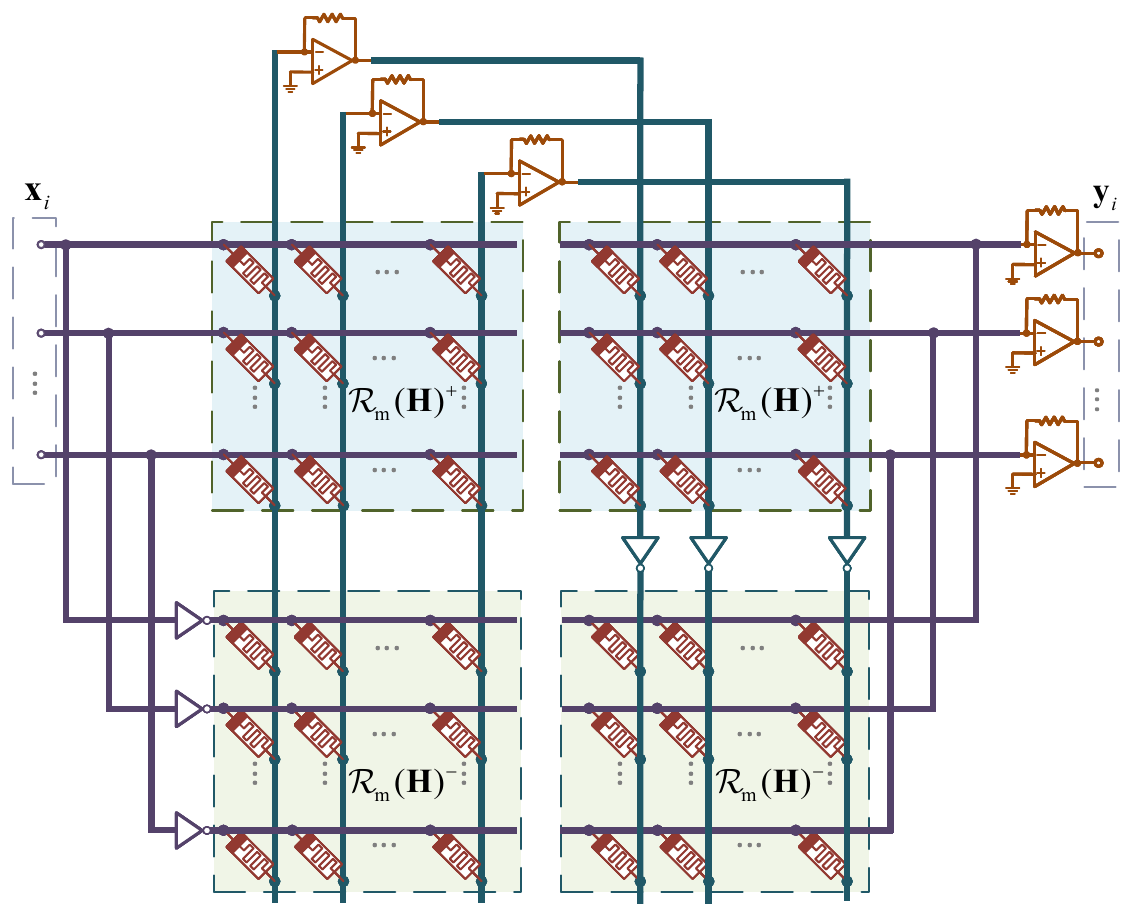}
\caption{The computation sub-circuit of the RRAM-based LQ decomposition circuit.}
\label{computation}
\end{figure}

The computation circuit contains $K$ identical sub-circuits in parallel. For convenience, the input and output voltage vector of the $i$th computation sub-circuit are denoted by
\begin{align}
    \mathbf{x}_i=\left[{\mathbf{x}_i^{(\text{R})}}^T,{\mathbf{x}_i^{(\text{I})}}^T\right]^T,\\
    \mathbf{y}_i=\left[{\mathbf{y}_i^{(\text{R})}}^T,{\mathbf{y}_i^{(\text{I})}}^T\right]^T,
\end{align}
where ${\mathbf{x}_i^{(\text{R})}},{\mathbf{x}_i^{(\text{I})},{\mathbf{y}_i^{(\text{R})}},{\mathbf{y}_i^{(\text{R})}}\in\mathbb{R}^{K\times1}}$. As shown in Fig. \ref{computation}, the $i$th computation sub-circuit perform
\begin{align}\label{y2x_real}
    \mathbf{y}_i=\mathcal R_{\mathrm m}(\mathbf{H})\mathcal R_{\mathrm m}(\mathbf{H})^T\mathbf{x}_i.
\end{align}
The output matrix $\hat{\mathbf{L}}$ is formed by
\begin{align}
    \hat{\mathbf{L}}=
    \left[\mathcal{R}_{\mathrm v}^{-1}(\mathbf{y}_{1}),\mathcal{R}_{\mathrm v}^{-1}(\mathbf{y}_{2}),\dots,\mathcal{R}_{\mathrm v}^{-1}(\mathbf{y}_{K})\right].
    \label{eq:rf_Lhat_readout}
\end{align}

The constraint circuit described next regulates $\mathbf{x}_i$ so that $\hat{\mathbf{L}}$ constitutes the desired result. First, selected ports of $\mathbf{x}_i$ are grounded. Define
\begin{align}
    \mathbf{P}=\left[\mathcal{R}_{\mathrm v}^{-1}(\mathbf{x}_{1}),\mathcal{R}_{\mathrm v}^{-1}(\mathbf{x}_{2}),\dots,\mathcal{R}_{\mathrm v}^{-1}(\mathbf{x}_{K})\right].
    \label{eq:rf_P_definition}
\end{align}
The upper-triangular structure of $\mathbf{P}$ is imposed by grounding
\begin{align}
    x_{i,j}^{(\mathrm R)}=x_{i,j}^{(\mathrm I)}=0,\qquad j>i,\label{eq:rf_x_ground}
\end{align}
where $x_{i,j}^{(\mathrm R)}$ and $x_{i,j}^{(\mathrm I)}$ denote the $j$th element of $\mathbf{x}_{i}^{(\mathrm R)}$ and $\mathbf{x}_{i}^{(\mathrm I)}$, respectively. In addition, the diagonal entries of $\mathbf{P}$ are set real-valued by grounding
\begin{align}
    x_{i,i}^{(\mathrm I)}=0.
\end{align}

The remaining ungrounded input ports are generated by an additional circuit. It is built from the sub-circuit as shown in Fig. \ref{constrain}. This sub-circuit utilizes OAs and multipliers to constrain the inner product of two voltage vectors to a prescribed value. For $\mathbf{x}_i$, we construct its constraints with $\mathbf{y}_j$ for $j\le i$
\begin{align}\begin{cases}
 \phi_{ij}^{(\text{R})}\triangleq(\mathbf{x}_i)^T\mathbf{y}_j=\delta_{ij},\\
    \phi_{ij}^{(\text{I})}\triangleq(\mathbf{J}_K\mathbf{x}_i)^T\mathbf{y}_j=0,
\end{cases}\label{eq:rf_real_constraints}
\end{align}
where $\delta_{ij}$ denotes the Kronecker delta and $\mathbf{J}_K$ is given by
\begin{align}
    \mathbf{J}_{K}
    =
    \begin{bmatrix}
        \mathbf{0}&-\mathbf{I}_{K}\\
        \mathbf{I}_{K}&\mathbf{0}
    \end{bmatrix}.
    \label{eq:pre_J_definition}
\end{align}

\begin{figure}[h]
\centering
\includegraphics[width=0.48\textwidth]{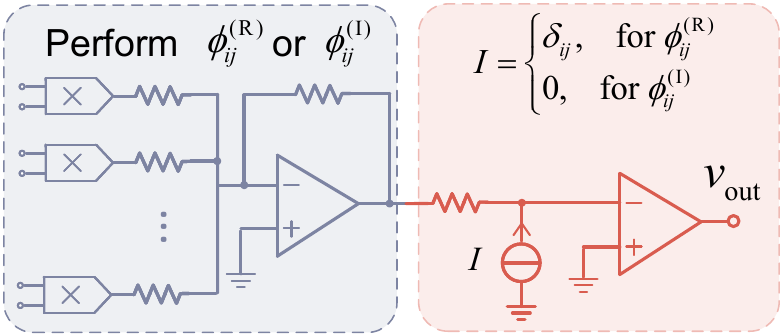}
\caption{The constraint sub-circuit of the RRAM-based LQ decomposition acceleration circuit.}
\label{constrain}
\end{figure}

We note that the voltage vector $\mathbf{J}_K\mathbf{x}_i$ is already available within the circuit, as the inverters are used in the computation circuit to realize the subtraction operation via $-\mathbf{x}_i$. Constraint sub-circuits in Eq. (\ref{eq:rf_real_constraints}) together represent one complex inner product as
\begin{align}\label{eq:rf_phi_complex_meaning}
\mathcal{R}_{\text{v}}^{-1}\left([\phi_{ij}^{(\text{R})},\phi_{ij}^{(\text{I})}]^T\right)=\left(\mathcal{R}_\text{v}^{-1}(\mathbf{x}_i)\right)^H\mathcal{R}_\text{v}^{-1}(\mathbf{y}_j)=\delta_{ij}.
\end{align}

For $j<i$, both $\phi_{ij}^{(\mathrm R)}$ and $\phi_{ij}^{(\mathrm I)}$ are implemented, and their OA outputs are fed back to provide $x_{i,j}^{(\mathrm R)}$ and $x_{i,j}^{(\mathrm I)}$, respectively. For $j=i$, only $\phi_{ii}^{(\mathrm R)}$ is implemented to provide $x_{i,i}^{(\mathrm R)}$ since $x_{i,i}^{(\mathrm I)}$ is grounded.
From Eq. (\ref{y2x_real}) and (\ref{eq:rf_phi_complex_meaning}), the whole constraint circuit imposes
\begin{align}
    \mathbf{P}^H\mathbf{H}\mathbf{H}^H\mathbf{P}=\mathbf{I}_K.
\end{align}
Combining the upper-triangular structure of $\mathbf{P}$, we obtain
\begin{align}
    \mathbf{P}=(\mathbf{L}^H)^{-1}\boldsymbol{\Theta}_{\pm},
\end{align}
where $\boldsymbol{\Theta}_{\pm}$ is a diagonal sign-ambiguity matrix, given by 
\begin{align}
\boldsymbol{\Theta}_{\pm}
=
\operatorname{diag}\{\theta_{1},\ldots,\theta_{K}\},
\qquad
\theta_i\in\{-1,+1\}.
\end{align}
Therefore, we can derive that output matrix $\hat{\mathbf{L}}$ is given by
\begin{align}
\hat{\mathbf{L}}=\mathbf{H}\mathbf{H}^{H}\mathbf{P}=\mathbf{L}\boldsymbol{\Theta}_{\pm}.
    \label{eq:rf_Lhat_solution}
\end{align}

Finally, $\mathbf{x}_{i}$ is routed to an additional RRAM-based MVM circuit storing $\mathcal{R}_{\mathrm m}(\mathbf{H}^{H})$ to perform
\begin{equation}
   \mathbf{t}_{i}=\mathcal{R}_{\mathrm m}(\mathbf{H}^{H})\mathbf{x}_{i},
    \label{eq:rf_ti}
\end{equation}
and $\hat{\mathbf{Q}}^H$ is formed by
\begin{equation}
    \hat{\mathbf{Q}}^{H}
    =
    \left[
    \mathcal{R}_{\mathrm v}^{-1}(\mathbf{t}_{1}),
    \ldots,
    \mathcal{R}_{\mathrm v}^{-1}(\mathbf{t}_{K})
    \right]=\mathbf{Q}^{H}\boldsymbol{\Theta}_{\pm}.\label{eq:rf_Qhat_solution}
\end{equation}

For standard LQ decomposition, the diagonal entries of $\mathbf{L}$ are usually chosen as positive real values. The sign ambiguity in $\boldsymbol{\Theta}_{\pm}$ does not change the lower-triangular structure of $\hat{\mathbf{L}}$ and the row orthogonality of $\hat{\mathbf{Q}}$. Therefore, the obtained matrices $\hat{\mathbf{L}}$ and $\hat{\mathbf{Q}}$ are equivalent for LQ decomposition and the THP algorithm.

\begin{figure}[h]
\centering
\includegraphics[width=0.45\textwidth]{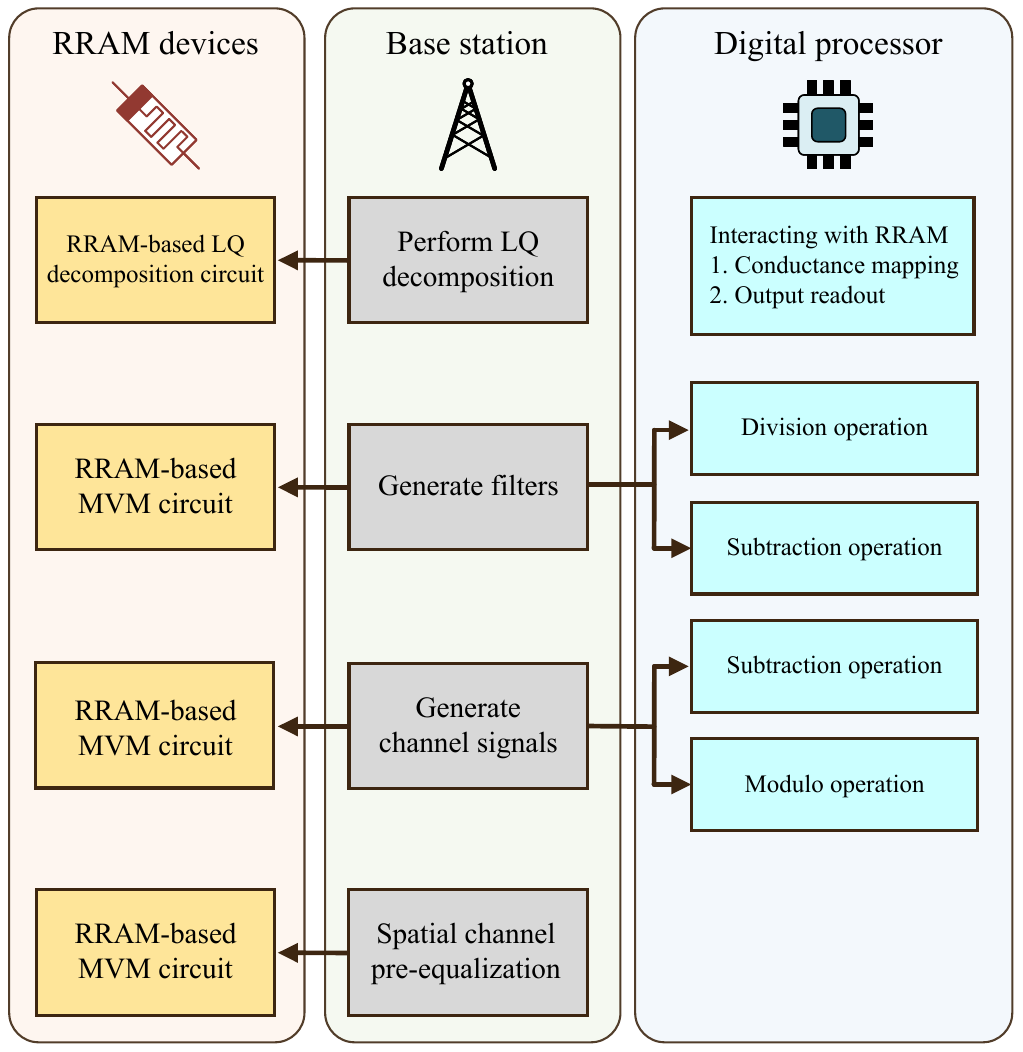}
\caption{System architecture of RRAM-based THP algorithm.}
\label{RRAM_THP}
\end{figure}

Fig. \ref{RRAM_THP} illustrates the system architecture of the RRAM-based THP algorithm. The RRAM devices are employed to perform LQ decomposition and MVM operations. The digital processors are utilized to interact with RRAM devices and handle simple, low-complexity computations. By leveraging RRAM devices to accelerate high-complexity operations, the RRAM-based THP algorithm achieves a significant reduction in computational complexity for base stations, which will be detailed in Sec.\uppercase\expandafter{\romannumeral5}.

% Except for the module operation in Eq. (\ref{mod}), 

\section{Impact of the Bit Precision}
In practical systems, the conductance of a memristor is quantized into a finite number of distinguishable states. This approach counteracts random errors caused by non-idealities by ensuring sufficient spacing between different conductance states\cite{liu2026memristor,Cao2020quantiRRAM}. In this case, quantization error becomes the dominant source of error and limits system performance. In this section, we analyze the SINR performance and achievable rate performance with the quantized RRAM devices. To derive the SINR performance, we first introduce a key preliminary finding.
\begin{lemma}\label{lemma1}
Let $\{x_i\}_{i=1}^t$ and $\{y_i\}_{i=1}^t$ be two sequences, each consisting of i.i.d. random variables. $\mathbb{E}(y_i)\neq0$. Define $X=\sum_{i=1}^{t} x_i$ and $Y=\sum_{i=1}^{t} y_i$. We obtain the following approximation,
\begin{equation}\label{lemma1_eq}
      \mathbb{E}\left(\frac{X}{Y}\right)\approx  \frac{\mathbb{E}(X)}{\mathbb{E}(Y)}.
\end{equation}
\end{lemma}
\begin{proof}
The proof is available in Appendix \ref{ProofLemma1}. 
\end{proof} 

Note that the proof of \textbf{Lemma 1} does not require any assumptions regarding the correlation between $x_i$ and $y_i$. The approximation in \textbf{Lemma 1} becomes more accurate when $t$ increases. In subsequent applications of \textbf{Lemma 1}, $t$ is equivalent to the number of antennas, which makes it highly precise in the massive MIMO system.

\subsection{Quantization error analysis}
In this subsection, we quantitatively analyze the impact of RRAM-induced quantization errors in communication systems. When the bit precision is limited, the matrix stored in the RRAM array is quantized. MVM operation is a linear operation that satisfies the distributive law.
By decomposing the matrix into lower-bit slices and accumulating the result, high-accuracy MVM is achievable for an RRAM circuit of low bit precision \cite{shafiee2016isaac}. 
However, for LQ decomposition, the inherently nonlinear nature prevents the application of this shift-and-add technique to enhance its precision. Consequently, the quantization error introduced by the RRAM-based LQ decomposition circuit cannot be readily mitigated. The quantization error analysis in this work focuses primarily on the RRAM-based LQ decomposition circuit.

To begin the analysis, we first define the quantization scheme adopted for memristor devices. Following the widely-adopted modeling approach in \cite{ZF_RRAM, RZF_RRAM}, we assume that memristive devices are uniformly quantized to obtain discrete and well-separated conductance states.
\begin{assumption}
Let $b$ represent the bit precision for a memristor. The conductance $G$ of a $b$-bit memristor takes $2^b-1$ uniform discrete states and a deep high resistance state
\begin{align}
    G=k\triangle G_0,\quad k\in\{0,1,\dots,2^b-1\},
\end{align}
where $\triangle G_0$ is a reference unit conductance.
\end{assumption}

The uniform quantization scheme adopted in \textbf{Assumption 1} is widely used to mitigate the impact of analog non-idealities \cite{li2018analogue,Cao2020quantiRRAM}. Notably, when employing robust and sufficiently distinguishable conductance states, quantization error becomes the dominant impairment. Accordingly, we restrict our attention to the impact of quantization error in this work.

Let $x$ represent the real part of the channel matrix element, or equivalently, the imaginary part of the channel matrix element. When using the method of Eq. (\ref{eq:Rm_definition}) to represent complex matrices, $x$ is equivalently quantized as
\begin{equation}\label{qx}
\bar{x}=\begin{cases}
nd&x\in((n-\frac12)d,(n+\frac12)d),\\(2^b-1)d&x>(2^b-\frac12)d,\\-(2^b-1)d&x<(-2^b+\frac12)d,\end{cases}
\end{equation}
where $n\in\{-2^b+1,\dots,2^b-1\}$, $\bar{x}$ denotes the quantized values of $x$ and $d$ denotes the quantization interval.

Next, we characterize the quantization error distribution by exploiting the statistical structure of Rayleigh fading channels. Let $\tilde{x}$ denote the quantization error of $x$, given by
\begin{equation}
    \tilde{x}=x-\bar{x}.
\end{equation}

Let $f_b(\tilde{x})$ denote the probability density function (PDF) of quantization errors when the bit precision of memristors is $b$. In \textbf{Lemma 2}, we deduce that $\tilde{x}$ is approximately uniformly distributed.

\begin{lemma}\label{lemma2}
When a $b$-bit RRAM device is utilized to store a value $x$ that follows a zero-mean Gaussian distribution with variance $\sigma^2$ and the quantization interval $d\in \left [\sigma/(2^b-1),\sigma\right]$, the quantization error $\tilde{x}$ is almost uniformed distributed in $[-\frac{d}{2},\frac{d}{2}]$. For $|\tilde{x}|<\frac{1}{2}d$, the relative discrepancy between $\tilde{x}$ and that of a uniformly distributed variable in PDF is upper bounded as
\begin{equation}
\frac{\left |\frac{1}{d}- f_b(\tilde{x})\right|}{\frac{1}{d}} \le \left(\mathrm{erfc}\left(\frac{(2^b-1)d}{\sqrt{2} \sigma}\right)+10^{-7}\right),
\end{equation}
where $\mathrm{erfc}(x)$ is complementary error function.
\end{lemma}

\begin{proof}
The proof is available in Appendix \ref{ProofLemma2}.
\end{proof}

Let the quantization interval $d$ be parameterized as
\begin{align}
    d=\frac{\zeta\sigma}{2^b},
\end{align}
where $\zeta>0$ is the scaling factor that controls the quantization range. From Eq. (\ref{qx}), the circuit is able to represent values within $[-\zeta\sigma, \,\zeta\sigma]$ in a quantized form. For a fixed bit precision $b$, a larger $\zeta$ covers a wider range of channel values, but at the cost of a larger quantization error within the range. Although no universal closed-form expression exists for the optimal scaling factor\cite{q_d}, the value of $\zeta$ can be optimized to minimize the mean squared error, as discussed in \cite{Sakr2021TSP}. In this work, we adopt the four-sigma rule of thumb to set $\zeta=4$ as a fixed and commonly used baseline\cite{foursigmaTIT}. This range covers approximately $99.994\%$ of channel values while maintaining a bounded quantization error. The analysis method can be extended to other values of $\zeta$. Under this setting, \textbf{Lemma 2} holds for all $b\ge 2$. Therefore, the relative discrepancy between $\tilde{x}$ and that of a uniformly distributed variable in PDF is upper bounded as
\begin{align}
\frac{\left |\frac{1}{d}- f_b(\tilde{x})\right|}{\frac{1}{d}}\le \mathrm{erfc}\left(\frac{(2^b-1)d}{\sqrt{2} \sigma}\right)+10^{-7}\nonumber\\\approx  \mathrm{erfc}\left(\frac{4}{\sqrt{2} }\right)+10^{-7}\le 10^{-4}.
\end{align}
Therefore, it is feasible to approximate the quantization error $\tilde{x}$ as a uniformed distributed variable in $[-\frac{d}{2},\frac{d}{2}]$. 

Let $\bar{\mathbf{H}}$ denote the quantized channel matrix and $\tilde{\mathbf{H}}$ denote the matrix of quantization error. When RRAM devices are employed to store the real part and the imaginary part of the element in the channel matrix, every quantized element $\bar{h}_{ij}$ in $\bar{\mathbf{H}}$ has its real part and imaginary part quantized as Eq. (\ref{qx}) with the quantization interval $d=4\sigma/2^b=2\sqrt{2}\sigma_h/2^b$. The matrix of quantization error $\tilde{\mathbf{H}}$ is given by
\begin{align}
\tilde{\mathbf{H}}=\mathbf{H}-\bar{\mathbf{H}}.
\end{align}

From \textbf{Lemma 2}, we can assume every element $\tilde{h}_{ij}$ in matrix $\tilde{\mathbf{H}}$ has its real part and imaginary part uniformly distributed in $[-\frac{d}{2},\frac{d}{2}]$.

\subsection{SINR performance analysis}

In this subsection, we consider the SINR performance of the RRAM-based THP algorithm. For simplicity, we set that the power allocated to all users is equal. From \cite{QuantizaedTHP}, a standard assumption for $\mathbf{w}$ is that the covariance of $\mathbf{w}$ is approximated as $\mathbf{R_w}=\frac{M}{M-1}\mathbf{I}_K$ for the $M$-ary QAM modulation. The precoding loss $\frac{M}{M-1}$ can be considered negligible as $M$ increases. Thus, we ignore the precoding loss and assume that $\mathbb{E}(\mathbf{s}^H\mathbf{s})=\mathbb{E}(\mathbf{w}^H\mathbf{w})$. 

First, we consider the ideal THP algorithm. Let $P_k$ denote the transmit power allocated to the $k$th user. From \cite{lkk}, the expectation of SNR of the $k$th user is given by
\begin{align}\label{withoutQ}
   \mathbb{E}\left (\text{SNR}_k \right )= \frac{(N-k+1)\sigma_h^2P_k}{N_0}.
\end{align}

Note that the SNR performance is affected by the order of the user’s channel vector in $\mathbf{H}$, i.e., $k$. We define $\gamma$ to represent the normalized SNR of the received signal that only encapsulates the effects of channel and transmit power, expressed as
\begin{align}
\label{gamma}
\gamma=\frac{\sigma^2_hP_k}{N_0}.
\end{align}

For clarity, we may use $\gamma$ as an equivalent alternative to the SNR of the received signal in the following analysis.

When the bit precision of the RRAM devices is limited, the proposed circuit performs LQ decomposition on the channel matrix represented by quantized conductance states. It should be noted that, in practical implementations, the computation results may further deviate from this quantized model due to circuit-level non-idealities\cite{Errorbuild}, such as wire parasitics and data converter imperfections. The modeling and compensation of these non-idealities have been investigated in relevant studies \cite{THU22RRAM,MinADC2026} and are also being explored in our follow-up work \cite{zyhOFDMtosubmit}. In this paper, we make idealized assumptions on these effects to focus on the impact of the finite bit precision of the RRAM devices. Accordingly, the RRAM circuit is modeled as performing
\begin{equation}
\bar{\mathbf{H}}=\mathbf{L}_{q}\mathbf{Q}_{q},
\end{equation}
where $(\cdot)_q$ denotes the computational results obtained by the quantized RRAM circuit. The received signal is written as
\begin{equation}
\begin{aligned}
\label{r_q}
\mathbf{r}=\mathbf{v} + \mathbf{G}_q  \tilde{\mathbf{H}}  \mathbf{F}_q \mathbf{w} + \mathbf{G}_q \mathbf{n}.
\end{aligned}
\end{equation}

Let $\rho_k$ denote the SINR for the $k$th user. From Eq. (\ref{r_q}), $\rho_k$ is expressed as \cite{QuantizaedTHP}
\begin{equation}
\rho_k=\frac{\gamma\bar{l}_{kk}^2}{\gamma\left\|\tilde{\mathbf{h}}_k\mathbf{F}_q\right \|_2^2+\sigma_h^2},
\end{equation}
where $\tilde{\mathbf{h}}_k$ denotes the $k$th row vector of $\tilde{\mathbf{H}}$ and $\bar{l}_{kk}$ denotes the element of $\mathbf{L}_q$ in row $k$ and column $k$. 
The SINR performance is limited by the quantization error.

Let $\rho_k^K$ denote the SINR of the $k$th user and $\mathbf{F}_{q,K}$ denote matrix $\mathbf{F}_{q}$ when the number of users is $K$. For the RRAM-based THP algorithm, as $K$ increases from $K_1$ to $K_2$ $(K_1<K_2)$, $\rho_k^{K_1}$ and $\rho_k^{K_2}$ have the following expression and relationship
\begin{equation}\label{lower}
\rho_k^{K_1}=\frac{\gamma\bar{l}_{kk}^2}{\gamma\left\|\tilde{\mathbf{h}}_k\mathbf{F}_{q,K_1}\right \|_2^2+\sigma_h^2}
    \ge
    \rho_k^{K_2}=\frac{\gamma\bar{l}_{kk}^2}{\gamma\left\|\tilde{\mathbf{h}}_k\mathbf{F}_{q,K_2}\right \|_2^2+\sigma_h^2}.
\end{equation}
The proof is available in Appendix \ref{Proofremark}.

It indicates that the presence of additional users leads to a decline in the SINR performance of the $k$th receiver. The SINR performance drops to its lower bound when the number of users reaches its maximum value, i.e., $K=N$. Its lower bound takes the form
\begin{equation}
\label{SINR_lower}
\breve{\rho}_k  = \frac{\gamma \bar{l}_{kk}^2}{\gamma \tilde{\mathbf{h}}_k\tilde{\mathbf{h}}_k^H+\sigma_h^2} .
\end{equation}

\begin{theorem}\label{Theorem1}
When the bit precision of RRAM devices is $b$, the expectation of SINR for the $k$th user is lower bounded as follows
\begin{equation}
\label{pro_eq}
        \mathbb{E}(\rho_k)  \ge \frac{(N-k+1)\gamma }{\frac{N}{3}\frac{\gamma}{2^{2b-2}} + 1} .
\end{equation}

The lower bound is attained when the number of users is equal to the number of antennas, i.e., $K=N$.
\end{theorem}
\begin{proof}
The proof is provided in Appendix \ref{Proof_pro1}.
\end{proof}

By \textbf{Theorem 1}, we reveal the impact of the normalized SNR of the received signal $\gamma$, the number of antennas $N$ and the bit precision $b$ on the SINR performance. Note that the approximation we employ becomes precise as $N$ and $b$ increase. The validity of these approximations is further supported by simulation results in Sec. \uppercase\expandafter{\romannumeral5}.

In the following analysis, we further explore the relationship between the SINR performance and these three key parameters. In \textbf{Corollary 1}, we first investigate the influence of the received SNR and analyze the asymptotic behavior at high SNR.

\begin{Corollary}\label{Corollary1}
Given a fixed $N$ and $b$, the lower bound of SINR performance exhibits a monotonic growth with respect to $\gamma$. 

There exists a ceiling effect such that when $\gamma$ is sufficiently high, the lower bound of SINR performance for the $k$th user approaches a constant as
\begin{equation}
    \lim _{\gamma \rightarrow \infty} \mathbb{E}(\breve{\rho}_k)  = 3\times2^{2b-2}\frac{N-k+1}{N}.
\end{equation}
\end{Corollary}
\begin{proof}
The first-order partial derivative of $\mathbb{E}(\breve{\rho}_k)$ with respect to $\gamma$ is expressed as
\begin{equation}
    \frac{\partial \mathbb{E}(\breve{\rho}_k)}{\partial \gamma}=\frac{9 \times 2^{4 b-4}(N-k+1)}{\left(N \gamma+3 \times 2^{2 b-2}\right)^{2}}>0.
\end{equation}

The ceiling effect is deduced by taking the limit of $\gamma$ approaching infinity in Eq. (\ref{pro_eq}).
\end{proof}

\textit{Remark 1:} \textbf{Corollary} \textbf{\ref{Corollary1}} characterizes the impact of the normalized SNR of the received signal $\gamma$ on the SINR performance and establishes the existence of a ceiling effect in the high-SNR regime. When the SNR is sufficiently high, quantization error emerges as the dominant limiting factor for SINR. Under this condition, each additional bit of precision yields a $6$ dB improvement in SINR.

Subsequently, we analyze the impact of the number of antennas $N$ and present the asymptotic SINR under large $N$ conditions.
\begin{Corollary}\label{Corollary2}
Given a fixed $\gamma$ and $b$, the lower bound of SINR performance exhibits a monotonic increase with respect to $N$. 

A ceiling effect is observed that when the number of antennas is sufficiently large, the lower bound of SINR performance for the $k$th user approaches a constant as
\begin{equation}
    \lim _{N \rightarrow \infty} \mathbb{E}(\breve{\rho}_k)  = 3\times2^{2b-2}.
\end{equation}
\end{Corollary}
\begin{proof}
The proof resembles that of \textbf{Corollary 1} and is therefore omitted for brevity.
\end{proof}

\textit{Remark 2:} \textbf{Corollary} \textbf{\ref{Corollary2}} demonstrates the impact of the number of antennas on the SINR performance and provides a closed-form expression for large $N$ cases. Notably, bit precision is the sole determinant of SINR under large $N$ conditions, where each additional bit results in a $6$ dB improvement in SINR. The received SNR no longer influences the SINR in this regime. This phenomenon arises because the array gain averages out the SINR disparity under different $\gamma$. Bit precision ultimately becomes the only factor that limits further SINR improvement.

For the impact of the bit precision on SINR, it is obvious that Eq. (\ref{pro_eq}) demonstrates monotonic improvement with increasing bit precision. There is a theoretical performance ceiling that as the bit precision approaches infinity, the SINR asymptotically converges to the expected ideal SNR without quantization error presented in Eq. (\ref{withoutQ}).

Higher bit precision entails increased costs in hardware fabrication and greater complexity in control, raising the question of what level of bit precision is appropriate \cite{mapping21IRPS}. The absolute SINR value alone cannot serve as a universal metric for evaluating bit precision under different system configurations. To provide a unified framework for evaluating the level of bit precision under different theoretically attainable SNRs, we introduce the quantization efficiency coefficient $\xi$:

\begin{definition}
The quantization efficiency coefficient $\xi$ is defined as the ratio between the SINR expectation with finite bit precision $b$ and the ideal expected SNR without quantization error:
\begin{equation}\label{definexi}
    \xi\triangleq\frac{\mathbb{E}(\rho)}{\mathbb{E}(\text{SNR})}.
\end{equation}
\end{definition}

The quantization efficiency coefficient $\xi$ reflects the degree of performance realization relative to the ideal case. The optimal bit precision is determined by the tolerance threshold for the quantization efficiency coefficient.

\begin{Corollary}\label{Corollary3}
(Optimal bit precision calculation) Given fixed $\gamma$ and $N$, the minimum required bit precision $b_0$ can be analytically determined by any specified quantization efficiency coefficient threshold $\xi_0$ as
\begin{equation}
\label{Coro3_xi}
    b_0=\left \lceil\frac{1}{2}\log_2\frac{4N\gamma\xi_0}{3(1-\xi_0)} \right \rceil.
\end{equation}
\end{Corollary}
\begin{proof}
From \textbf{Theorem} \textbf{\ref{Theorem1}} and Eq. (\ref{definexi}), the quantization efficiency coefficient is lower bounded as
\begin{equation}
    \xi_0\ge\frac{\mathbb{E}(\breve{\rho}_k)}{\mathbb{E}(\text{SNR}_k)}=\frac{1}{\frac{N}{3}\frac{\gamma}{2^{2b-2}} + 1} .
\end{equation}

Eq. (\ref{Coro3_xi}) is derived by enforcing the theoretical lower bound to exceed the minimum acceptable quantization efficiency coefficient.
\end{proof}

\textit{Remark 3:} \textbf{Corollary} \textbf{\ref{Corollary3}} provides a quantitative guideline for determining the optimal bit precision to maintain system performance under varying SNR conditions and antenna configurations. Besides, Eq. (\ref{Coro3_xi}) reveals a design principle that for every 6 dB increase in the SNR of the received signal or a quadrupling of the number of antennas necessitates an additional bit precision to obtain equivalent quantization efficiency.

\subsection{Achievable rate analysis}

Leveraging the derivation of SINR in the previous section, we investigate the ergodic achievable rate in the following analysis.

From Eq. (\ref{r_q}), the ergodic achievable rate for the $k$th user is given by
\begin{equation}
    R_k=\mathbb{E}\left (\log_2\left(1+\frac{\gamma\bar{l}_{kk}^2}{\gamma\left\|\tilde{\mathbf{h}}_k\mathbf{F}_q\right \|_2^2+\sigma_h^2}\right)\right ).
\end{equation}

\begin{theorem}\label{Theorem2}
When the bit precision of RRAM devices is $b$, the ergodic achievable rate for the $k$th user is lower bounded as follows
\begin{equation}
\label{pro2_eq}
     R_k=\log_2\left(1+\frac{(N-k+1)\gamma}{\frac{N}{3}\frac{\gamma}{2^{2b-2}}}\right).
\end{equation}

Similarly, the lower bound is achieved when the number of users matches the number of antennas, i.e., $K=N$.
\end{theorem}
\begin{proof}
Please see Appendix \ref{Proof_pro2}.
\end{proof}

\textbf{Theorem} \textbf{\ref{Theorem2}} demonstrates the impact of the normalized SNR of the received signal $\gamma$, the number of antennas $N$ and the bit precision $b$ on the RRAM-based THP algorithm from the perspective of ergodic rate. Similar to the \textbf{Theorem} \textbf{\ref{Theorem1}}, the approximation is precise as $N$ and $b$ improve. The accuracy of these approximations is further validated by simulation results in Sec. \uppercase\expandafter{\romannumeral5}.

Based on \textbf{Theorem 2}, we proceed to analyze the impact of the received SNR and the number of antennas and provide the asymptotic limit.

\begin{Corollary}\label{Corollary4}
Given a fixed $N$ and $b$, the lower bound of the ergodic achievable rate demonstrates a monotonic growth with respect to $\gamma$. There exists a ceiling effect such that when $\gamma$ is sufficiently high, the lower bound of the ergodic achievable rate for the $k$th user approaches a constant as
\begin{align}
\label{coro4_eq}
    \lim _{\gamma \rightarrow \infty} \breve{R}_k  &= \log_2\left(1+3\times2^{2b-2}\frac{N-k+1}{N}\right)\\
    &\approx 2b+\log_2\left(\frac{3(N-k+1)}{4N}\right).
\end{align}
\end{Corollary}
\begin{proof}
The monotonicity of Eq. (\ref{pro2_eq}) is derived through the monotonicity of the term $\frac{(N-k+1)\gamma}{\frac{N}{3}\frac{\gamma}{2^{2b-2}}}$ proven in \textbf{Corollary} \textbf{\ref{Corollary1}}. Eq. (\ref{coro4_eq}) is obtained by taking the limit of $\gamma$ approaching infinity in Eq. (\ref{pro2_eq}). 
\end{proof}

\begin{Corollary}\label{Corollary5}
Given a fixed $\gamma$ and $b$, the lower bound of the ergodic achievable rate exhibits a monotonic increase with respect to $N$. There exists a ceiling effect such that when $N$ is sufficiently large, the lower bound of the ergodic achievable rate for the $k$th user approaches a constant as
\begin{equation}
\label{coro5_eq}
    \lim _{N \rightarrow \infty} \breve{R}_k  = \log_2\left(1+3\times2^{2b-2}\right)\approx 2b+\log_2\left(\frac{3}{4}\right).
\end{equation}
\end{Corollary}
\begin{proof}
The proof resembles that of \textbf{Corollary 4} and is therefore omitted for brevity.
\end{proof}

\textit{Remark 4:} From \textbf{Corollary} \textbf{\ref{Corollary4}} and \textbf{\ref{Corollary5}}, the lower bound of the ergodic achievable rate exhibits an approximately linear growth trend with respect to the bit precision $b$ under conditions of high SNR or large $N$. Similar to \textbf{Corollary} \textbf{\ref{Corollary2}}, when the number of antennas is sufficiently large, the SNR of the received signal does not influence the lower bound of the ergodic achievable rate. The number of antennas $N$ mitigates the performance disparity under varying $\gamma$ levels.

\section{Simulation Result}
In this section, we present the simulation results to validate the feasibility and accuracy of the RRAM circuit and our theoretical analysis of bit precision.
\subsection{RRAM Circuit Result}
\begin{figure}[h]
\centering
\includegraphics[width=0.44\textwidth]{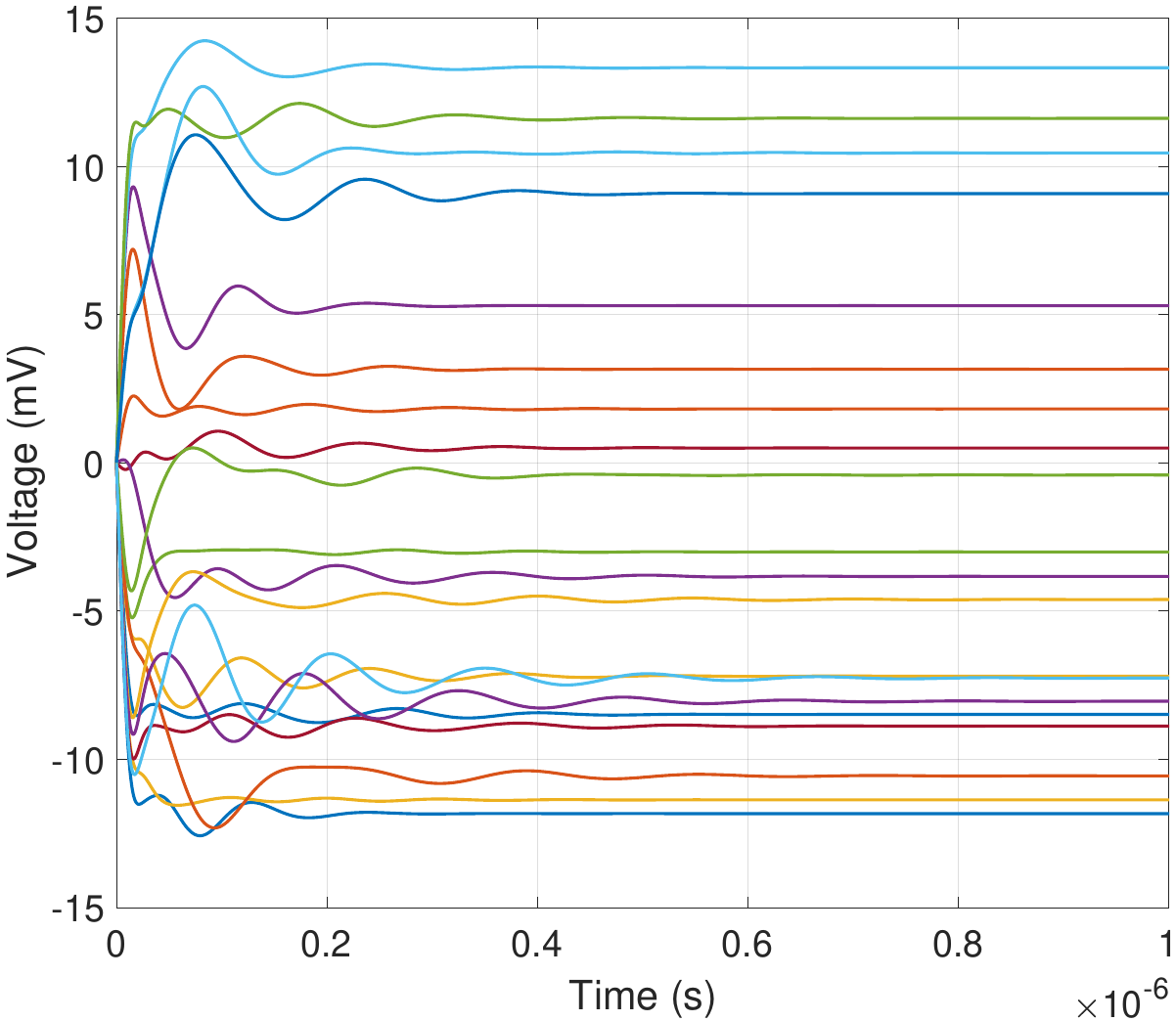}
\caption{Transient output voltage results of $\mathbf{L}$ from the RRAM circuit. $N=128$, $K=32$.}
\label{transi}\end{figure}

To validate the feasibility of the proposed circuit, we adopt a mixed-level simulation approach to implement the RRAM-based LQ decomposition circuit in LTspice. The constraint sub‑circuits are built using device‑level modeling to ensure functional accuracy, and the memristor crossbar arrays are modeled with a high‑level behavioral model to capture the MVM operation. OAs are set to have a gain-bandwidth product of $100$ MHz. The simulation is performed for a representative large‑scale MIMO configuration with $N=128, K=32$. Representative transient results of output are depicted in Fig. \ref{transi}. The output of the RRAM circuit converges rapidly in less than $1$ $\mu \text{s}$. 

\begin{figure}[h]
\centering
\includegraphics[width=0.44\textwidth]{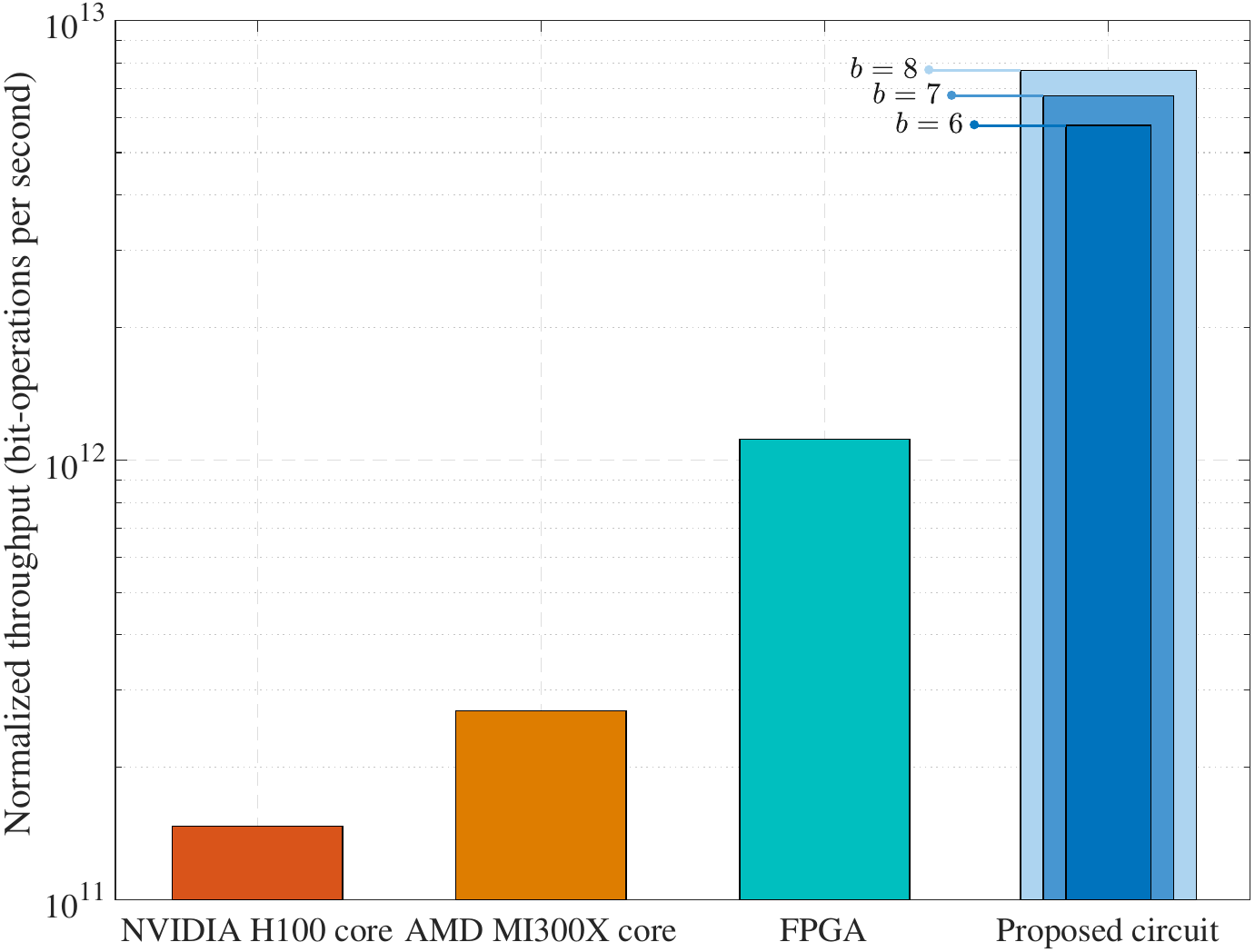}
\caption{Comparison of computing throughput for a $128\times32$ MIMO system.}
\label{compare_computation}\end{figure}

To benchmark the computational performance, we compare the proposed circuit with commercial GPUs (NVIDIA H100\cite{nvidia2024H100} and AMD MI300X\cite{amd2025MI300x}) and FPGA\cite{FPGA2016QR}. For a fair comparison across platforms with different arithmetic precisions, we follow the bit-normalization concept of\cite{BitNorm2022} and define the throughput metric as
\begin{align}
    \eta_{\text{bit}} = \frac{B_{\text{rep}} N_{\text{op}}}{T},
\end{align}
where $N_{\text{op}}$ denotes the number of operations completed within execution time $T$, and $B_{\text{rep}}$ is the bit-width of the numerical representation used by the corresponding implementation. The resulting metric is expressed in bit-operations per second. For the proposed RRAM circuit, $N_{\text{op}}$ is the number of digital operations required by the corresponding LQ decomposition. $B_{\text{rep}}$ is taken as the bit precision $b$ since the circuit performs decomposition on matrices with an effective bit-width of $b$. The throughput of the commercial GPU is normalized to that of a single core\cite{zuo2025precise}. As shown in Fig. \ref{compare_computation}, the proposed circuit achieves a normalized throughput that significantly outperforms both the GPU and FPGA implementations. The result highlights its substantial advantage for real‑time precoding in massive MIMO systems.

\begin{figure}[h]
\centering
\includegraphics[width=0.44\textwidth]{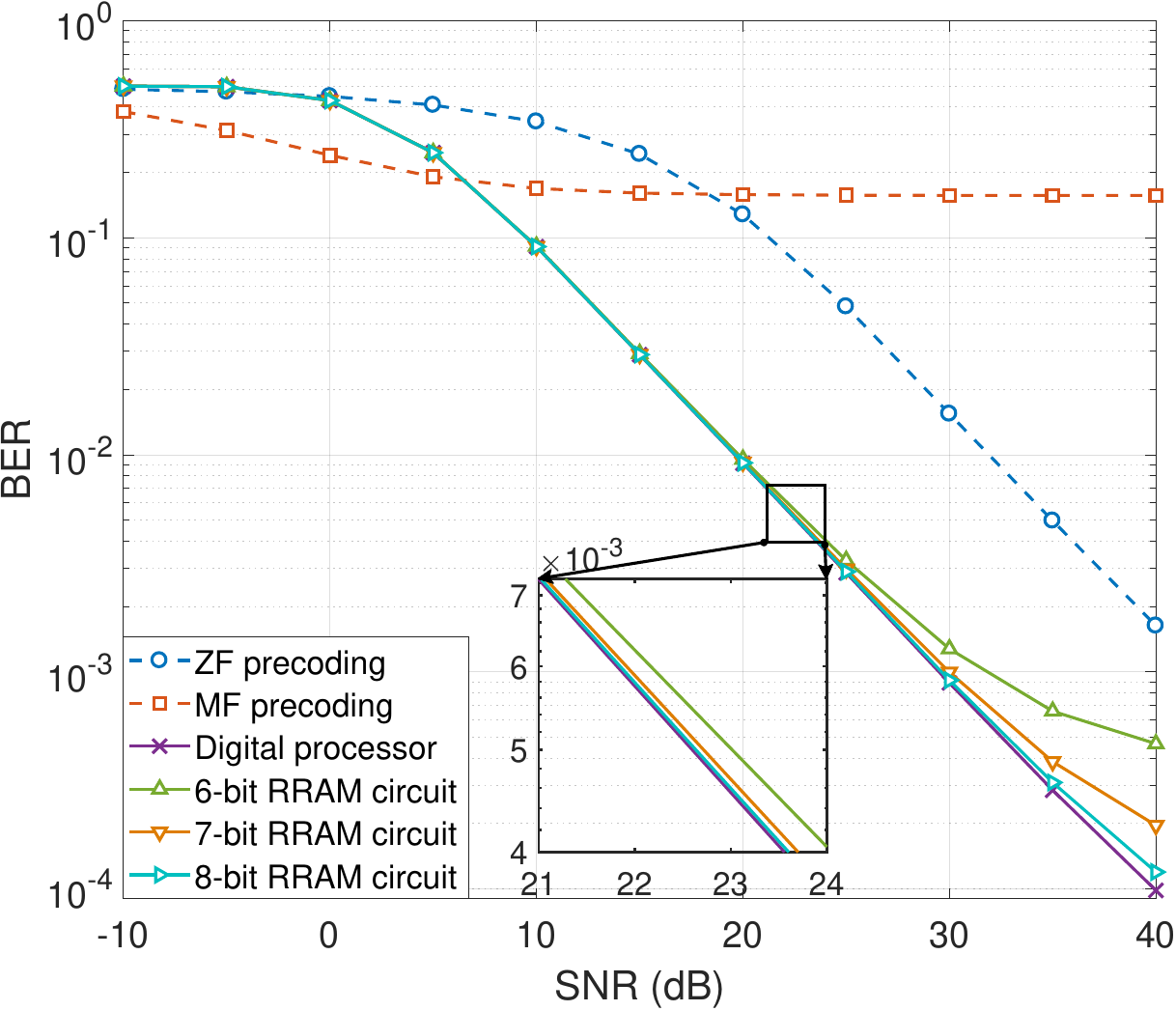}
\caption{BER performance of the RRAM-based THP algorithm. $N=32$, $K=32$.}
\label{BER}\end{figure}

To validate the reliability of the proposed RRAM-based THP algorithm, we evaluate its bit error rate (BER) performance under finite-bit precision. The results are compared with two representative linear precoding schemes. As depicted in Fig. \ref{BER}, the RRAM-based THP algorithm exhibits negligible performance degradation relative to that of the digital processor. Noticeable differences in BER performance between the RRAM circuit and the digital processor arise only at sufficiently high SNR levels.
Besides, the RRAM-based THP algorithm under finite-bit constraints still achieves significantly better BER performance than the linear precoding schemes, proving the advantages of the proposed architecture.

\begin{table}[h]
\caption{The computational complexity comparison}
\label{table I}
\centering
\scalebox{0.9}{
\begin{tabular}{|>{\centering\arraybackslash}m{6em}|>{\centering\arraybackslash}m{8.5em}|>{\centering\arraybackslash}m{6.3em}|}
\hline
Operation & Computational complexity of digital processors\cite{THPtime} & Computational complexity of RRAM circuits\\
\hline
LQ decomposition & $8K^2N-8K^3/3$ & $\boldsymbol{1}$ \\
\hline
Obtain filters & $8K^3-2K^2+2K$ & $\boldsymbol{4K}$ \\
\hline
Generate $\mathbf{w}$ & $4K^2+4K-8$ & $\boldsymbol{11K-11}$\\
\hline
Generate $\mathbf{z}$ & $8KN-2N$ & $\boldsymbol{1}$\\
\hline
$\textbf{Total}$ & $16K^3/3+8K^2N+2K^2+8KN+6K-2N-8$ & $\boldsymbol{15K-9}$\\
\hline
\end{tabular}}
\end{table}

The computational complexity of the RRAM-based THP algorithm is summarized in Table I. In contrast to the traditional digital processor, the complexity order of the RRAM-based THP algorithm is reduced from $\mathcal{O}(K^3)+\mathcal{O}(K^2N)$ to a linear scale of $\mathcal{O}(K)$. It implies the significant advantage of employing RRAM devices to implement nonlinear THP algorithms.

\subsection{Quantization Result}

Then, we depict the SINR performance and achievable rate under different normalized SNR of the received signal $\gamma$, the number of antennas $N$, bit precision $b$ to validate our analysis in Sec. \uppercase\expandafter{\romannumeral4}. Note that in the subsequent numerical simulations, we consider the scenario where the number of users equals the number of antennas ($K = N$). This configuration is chosen to verify the accuracy of the closed-form expression derived in \textbf{Theorem} \textbf{\ref{Theorem1}} and \textbf{\ref{Theorem2}} when the system performance drops to its lower bound.

\begin{figure}[h]
\centering
\includegraphics[width=0.44\textwidth]{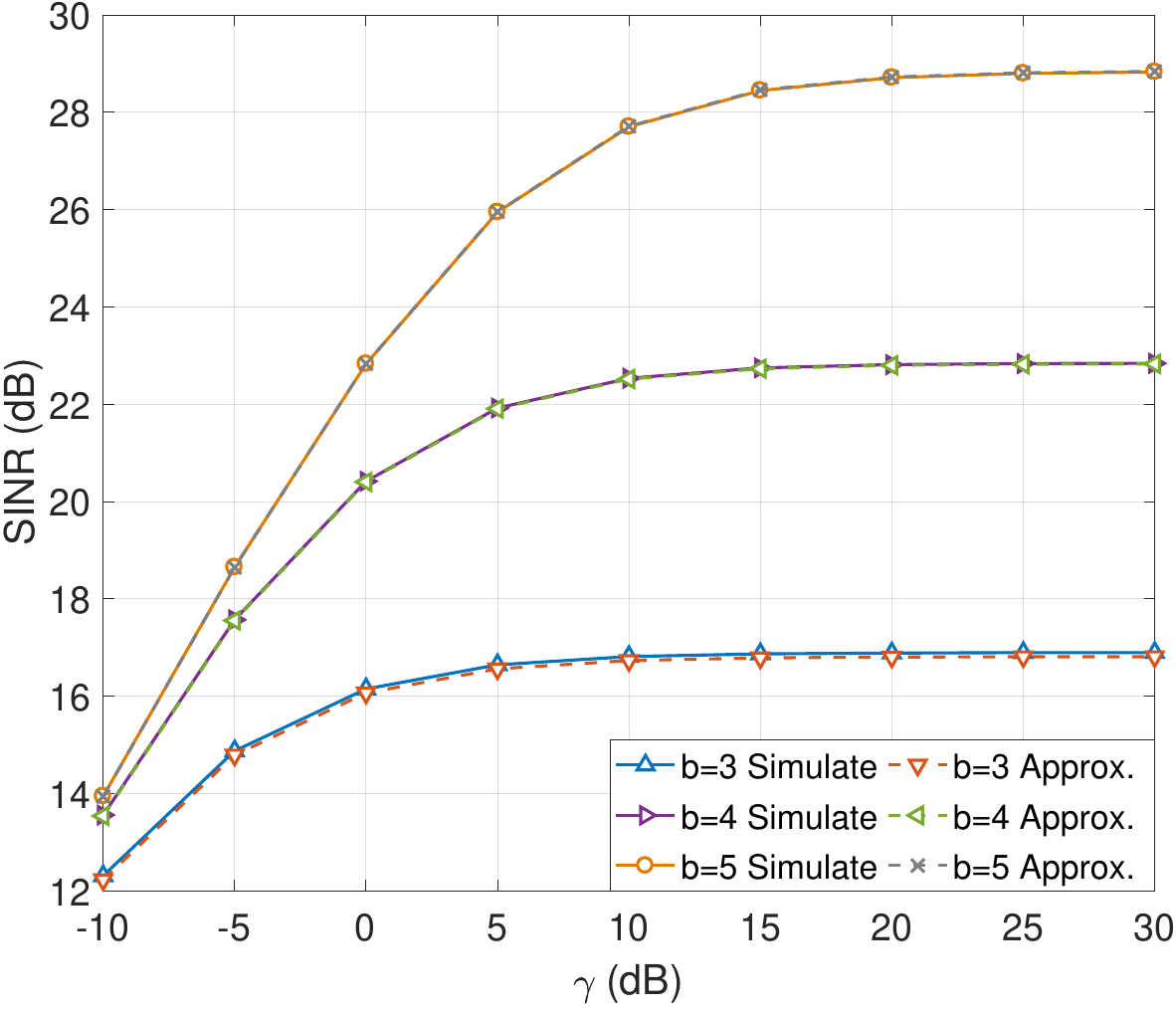}
\caption{The variation of the SINR performance with respect to $\gamma$. $N=256$. (The term 'approx.' is an abbreviation for 'approximation'.)}
\label{gamma_SINR}
\end{figure}

As shown in Fig. \ref{gamma_SINR}, we plot the SINR performance under different $\gamma$ and $b$. The simulation result and approximate result are obtained by Monte Carlo method and Eq. (\ref{pro_eq}), respectively. From Fig. \ref{gamma_SINR}, simulation results closely match the approximation, demonstrating the accuracy of our analysis in Sec. \uppercase\expandafter{\romannumeral4}. Besides, the SINR behavior aligns precisely with \textbf{Corollary} \textbf{\ref{Corollary1}}, showing monotonic growth with $\gamma$ and eventually converging to a ceiling. At high SNR, the SINR exhibits a characteristic $6$ dB improvement per additional bit precision.

\begin{figure}[h]
\centering
\includegraphics[width=0.44\textwidth]{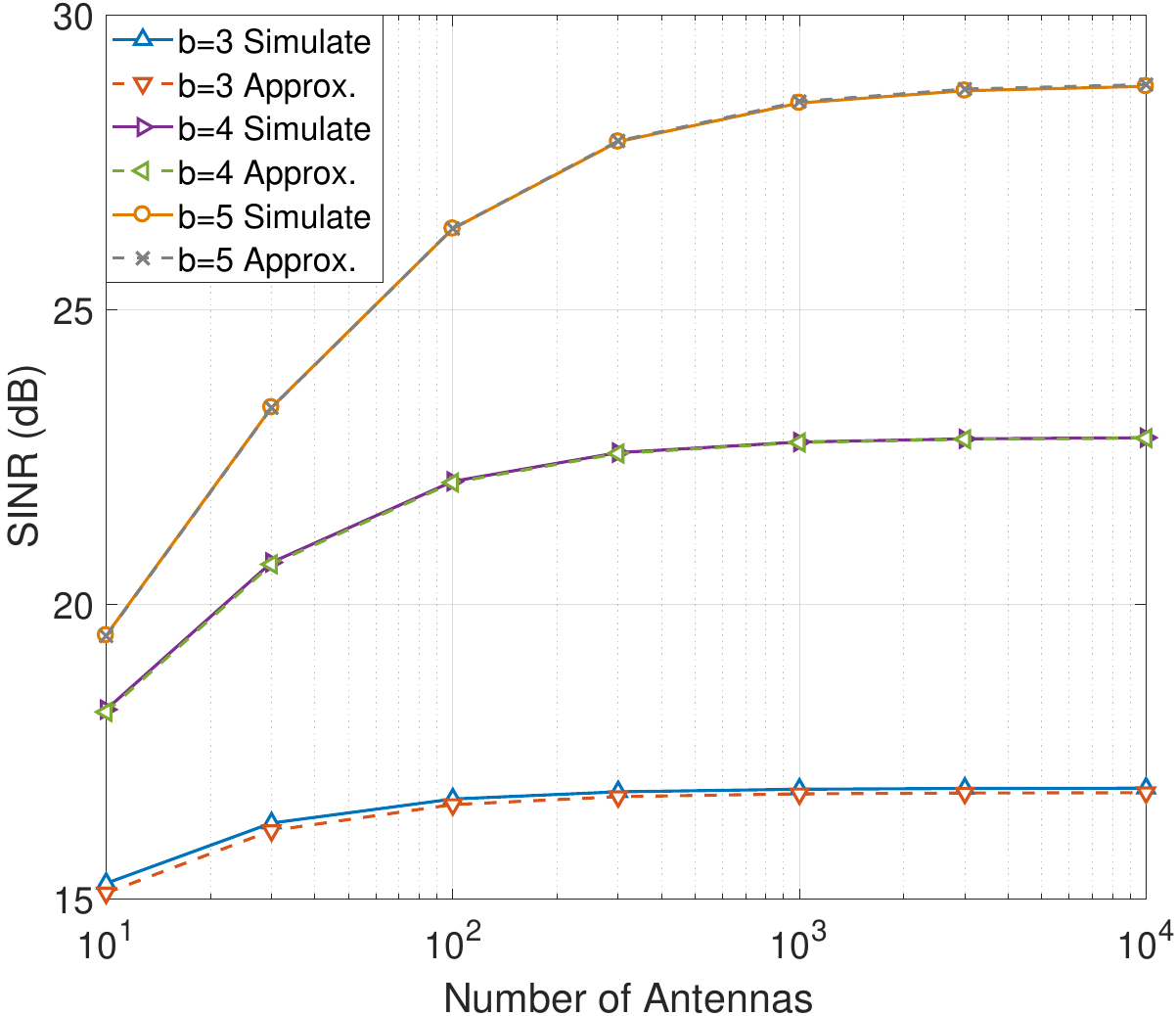}
\caption{The variation of the SINR performance with respect to the number of antennas. $\gamma=10$ dB.}
\label{N_SINR}
\end{figure}

Fig. \ref{N_SINR} shows the relationship between the SINR and the number of antennas $N$. The SINR performance demonstrates monotonic improvement with increasing $N$ before converging to a constant determined by the bit precision $b$. Fig. \ref{N_SINR} also confirms the characteristic $6$ dB/bit improvement predicted by theoretical analysis at large $N$ conditions. The simulation result provides strong experimental validation for \textbf{Corollary} \textbf{\ref{Corollary2}}.

\begin{figure}[h]
\centering
\includegraphics[width=0.44\textwidth]{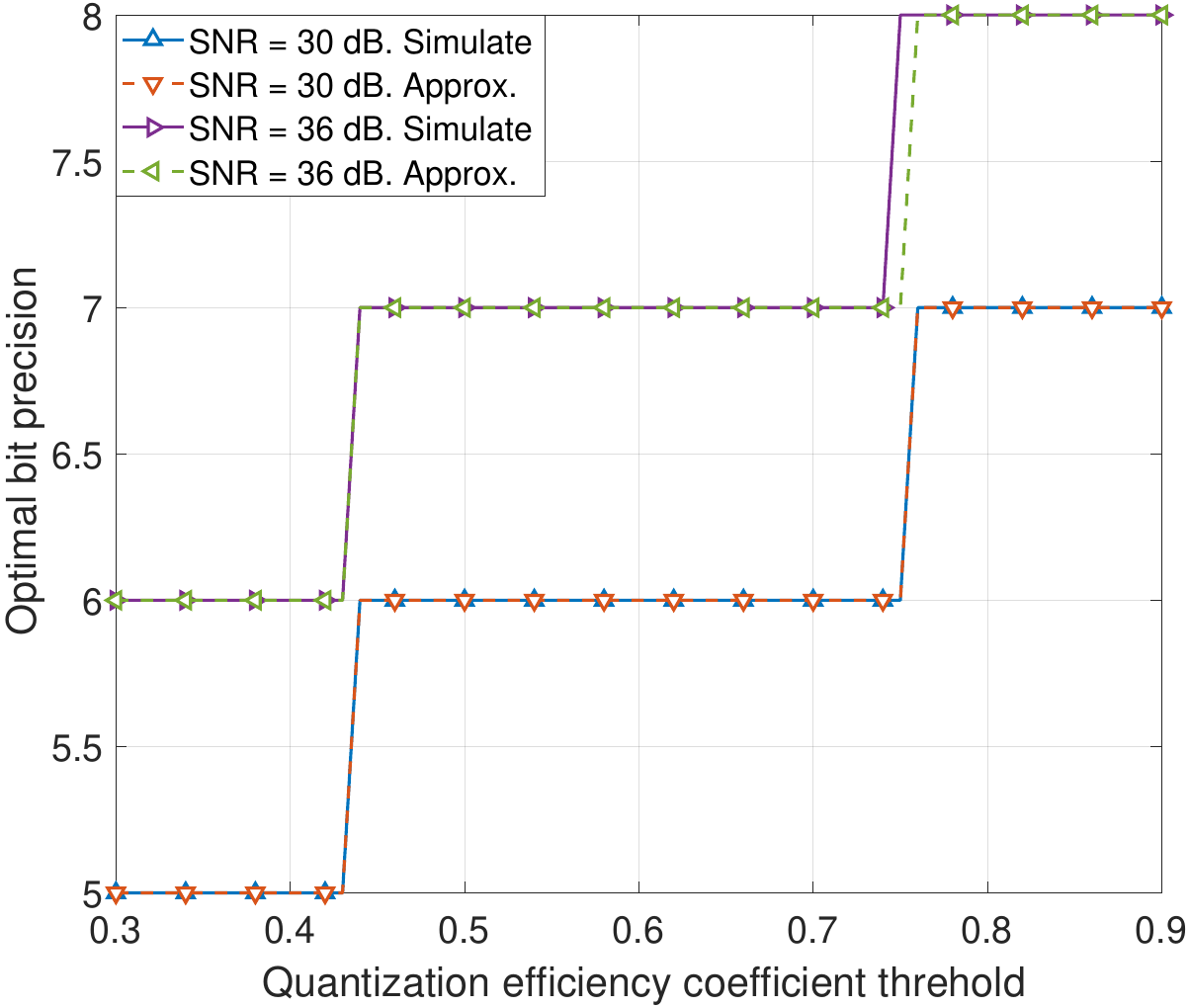}
\caption{The variation of the optimal bit precision with respect to the tolerable threshold of the quantization efficiency coefficient. $N=256$.}
\label{xi_bit}
\end{figure}

For various thresholds of acceptable quantization efficiency coefficients $\xi_0$, we calculate the optimal bit precision. In Fig. \ref{xi_bit}, the simulation result is obtained by sweeping bit precision in ascending order and employing Monte Carlo trials at each precision level to compute SINR until satisfying the requirement of $\xi_0$. The approximate result is calculated by Eq. (\ref{Coro3_xi}). As shown in Fig. \ref{xi_bit}, the closed-form solution derived in \textbf{Corollary} \textbf{\ref{Corollary3}} demonstrates remarkable consistency with exhaustive numerical search results. The results also imply that a $7$-bit RRAM array is sufficient to achieve over $90\%$ SINR performance when the SNR of the received signal is below $30$ dB and the number of antennas $N$ is below $256$.

\begin{figure}[h]
\centering
\includegraphics[width=0.44\textwidth]{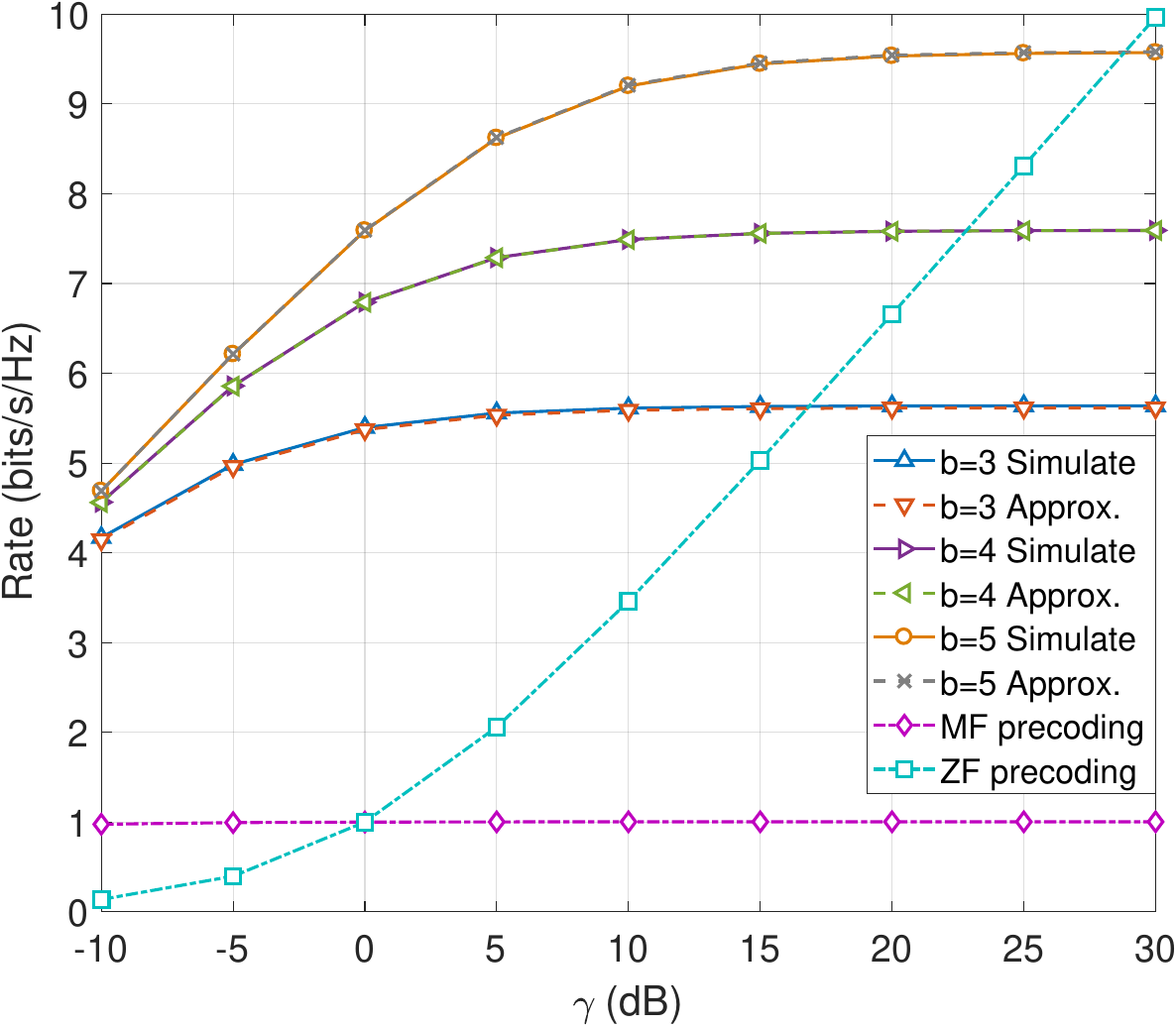}
\caption{The variation of the achievable rate with respect to $\gamma$. $N=256$.}
\label{gamma_R}
\end{figure}

\begin{figure}[h]
\centering
\includegraphics[width=0.44\textwidth]{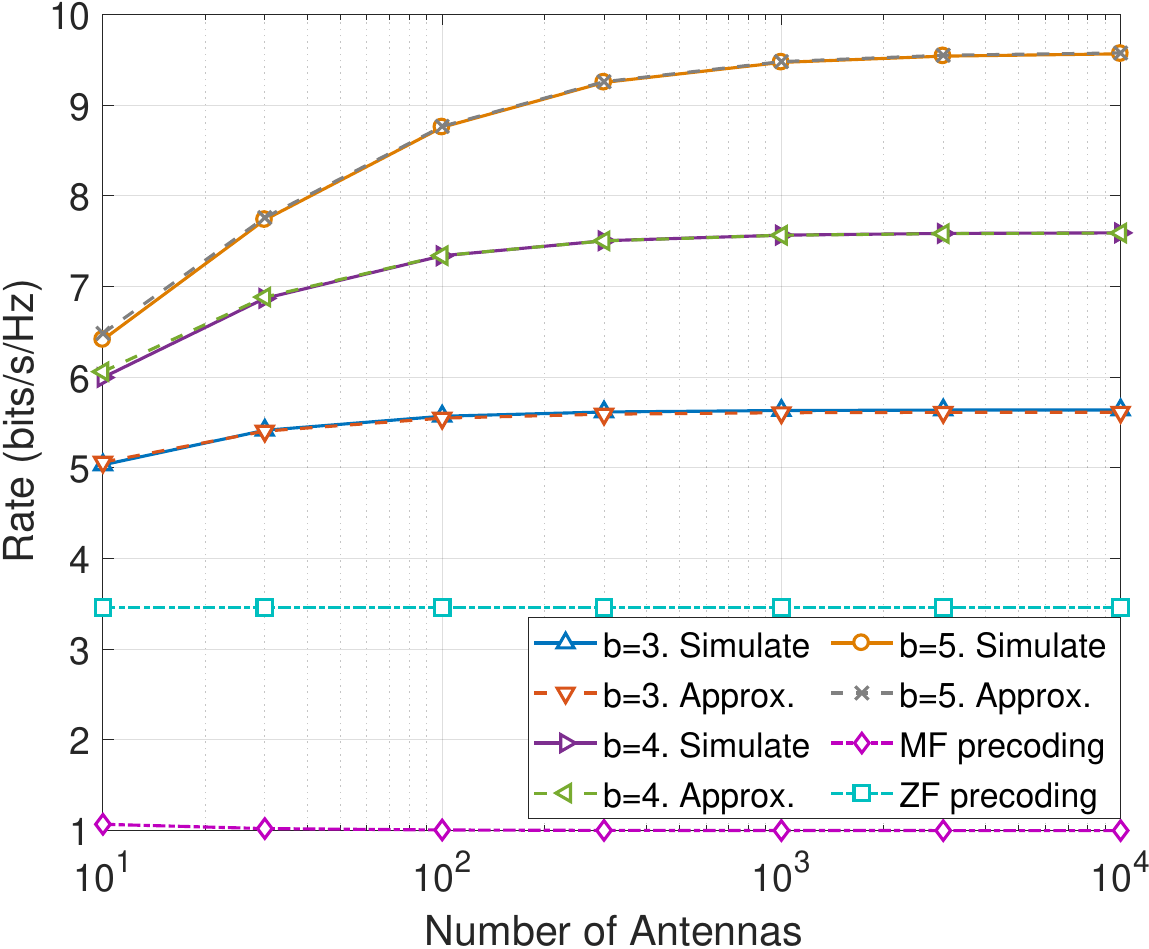}
\caption{The variation of the achievable rate with respect to the number of antennas. $\gamma=10$ dB.}
\label{N_R}
\end{figure}

Fig. \ref{gamma_R} and \ref{N_R} illustrate the ergodic rate as functions of $\gamma$ and $N$, respectively. The rates of the MF precoder and ZF precoder are obtained from \cite{tan2018multiuser}. Across the vast majority of SNR regimes and all antenna configurations considered, the proposed RRAM-based THP outperforms both ideal MF and ZF precoding. The ideally implemented ZF precoder may achieve comparable performance when the RRAM-based THP rate approaches the ceiling imposed by limited precision at sufficiently high SNRs. It further highlights the practical importance of the proposed quantization efficiency coefficient to determine the optimal bit precision across diverse system parameters. Besides, the convergence ceiling of the RRAM-based THP rate exhibits a linear dependence on bit precision under conditions of high SNR or large $N$, as predicted by \textbf{Corollary 4} and \textbf{5}.

\section{Conclusion}
In this paper, we proposed an IMC architecture that utilized RRAM circuits to reduce the computational complexity of the nonlinear THP algorithm to a linear scale. A computation-constraint principle to perform nonlinear matrix operations by RRAM circuits and an RRAM-based LQ decomposition acceleration circuit were proposed. Our results imply the potential of RRAM for implementing computationally intensive algorithms in practical systems. Besides, considering the quantization error caused by the limited bit precision of memristors, we analyzed the impact of bit precision on the SINR and achievable rate. We identified the ceiling effect for both the number of antennas and SNR and revealed the characteristic that under conditions of high SNR or large $N$, every $1$ bit increase in bit precision brings a $6$ dB improvement in SINR and linear growth in achievable rate. To determine the optimal bit precision for practical application, we defined the quantization efficiency coefficient as a quantitative guideline for varying SNRs and antenna configurations. Simulation results substantiated the practicality and accuracy of both the RRAM circuit and our theoretical analysis.

\appendix
\subsection{Proof of Lemma 1}\label{ProofLemma1}

Let $\mu_X$ and $\mu_Y$ denote the expectation of $X$ and $Y$, respectively. For function $g(X,Y)=\frac{X}{Y}$, by applying Taylor series expansion of $g(X,Y)$ around $\mu=(\mu_X,\mu_Y)$, we get that
\begin{align}\label{Taylor}
&g(X,Y)\approx \left.g\right|_{\mu}+\left.\frac{\partial g}{\partial X}\right|_{\mu}\left(X-\mu_{X}\right)+\left.\frac{\partial g}{\partial Y}\right|_{\mu}\left(Y-\mu_{Y}\right)\nonumber
\\&+\frac{1}{2}\left[\left.\frac{\partial^{2} g}{\partial X^{2}}\right|_{\mu}\left(X-\mu_{X}\right)^{2}+\left.2 \frac{\partial^{2} g}{\partial X \partial Y}\right|_{\mu}\left(X-\mu_{X}\right)\left(Y-\mu_{Y}\right)\right.\nonumber\\
&\left.+\left.\frac{\partial^{2} g}{\partial Y^{2}}\right|_{\mu}\left(Y-\mu_{Y}\right)^{2}\right].
\end{align}

Higher-order terms are neglected. From Eq. (\ref{Taylor}), $\mathbb{E}\left(\frac{X}{Y}\right)$ takes the form
\begin{align}\label{exofxy}
    \mathbb{E}\left(\frac{X}{Y}\right) \approx\frac{\mu_X}{\mu_Y}-\frac{\text{Cov}(X,Y)}{\mu_Y^2}+\frac{2\mu_X}{\mu_Y^3}\operatorname{Var}(Y),
\end{align}
where $\text{Cov}(X,Y)$ denotes the covariance of $X$ and $Y$. $\operatorname{Var}(Y)$ denotes the variance of $Y$. Let $\mu_x$ and $\mu_y$ represent the expectation of $x_i$ and $y_i$, respectively. $\sigma_y^2$ denotes the variance of $y_i$. Eq. (\ref{exofxy}) is further given by
\begin{align}
    \mathbb{E}\left(\frac{X}{Y}\right) \approx\frac{\mu_X}{\mu_Y}-\frac{\text{Cov}(x,y)}{t\mu_y^2}+\frac{2\mu_x\sigma_y^2}{t\mu_y}
\end{align}
When $t$ is large, we obtain the following approximation
\begin{align}
    \mathbb{E}\left(\frac{X}{Y}\right)\approx\frac{\mu_X}{\mu_Y}=\frac{\mathbb{E}(X)}{\mathbb{E}(Y)}
\end{align}

The approximation is more accurate as $t$ grows.

\subsection{Proof of Lemma 2}\label{ProofLemma2}
In this appendix, we investigate the distribution of the quantization error $\tilde{x}$ when a finite-bit memristor stores Rayleigh fading channel coefficients. We prove that the PDF of $\tilde{x}$ approximates a variable uniformly distributed at $\left[-\frac{d}{2},\frac{d}{2} \right ]$ when $d\in \left [\sigma/(2^b-1),\sigma\right]$.

The PDF of quantization error $f_b(\tilde{x})$ has the following expression
\begin{equation}
f_b(\tilde{x})=\begin{cases}\frac{1}{\sqrt{2\pi}\sigma}\sum_{k=-2^b+1}^{2^b-1}e^{-\frac{(kd+\tilde{x})^2}{2\sigma^2}}&,|\tilde{x}|\le\frac{1}{2}d,
\\\\
\frac{1}{\sqrt{2\pi}\sigma}e^{-\frac{ \left( \left( 2^b-\frac{1}{2} \right) d+|\tilde{x}| \right) ^2}{2\sigma^2}}&,|\tilde{x}|>\frac{1}{2}d.\end{cases}
\end{equation}

When the bit precision is infinite, the according PDF $f_{\infty}(\tilde{x})$ is written as
\begin{align}
f_{\infty}(\tilde{x})=\frac{1}{\sqrt{2\pi}\sigma}\sum_{k=-{\infty}}^{{\infty}}e^{-\frac{(kd+\tilde{x})^2}{2\sigma^2}}\nonumber\\=\frac{1}{d}(1+\delta_{\infty}(\tilde{x})),\quad|\tilde{x}|<\frac{1}{2}d.
\end{align}
where $\delta_{\infty}(\tilde{x})/d$ is the sum of the high-frequency components of the Fourier expansion of $f_{\infty}(\tilde{x})$. The first term $\frac{1}{d}$ is the DC component of $f_{\infty}(\tilde{x})$, which is equivalent to a uniformly distributed variable. From \cite{Uniform}, when $d\le \sigma$, the sum of the high-frequency components $\left | \delta_{\infty}(\tilde{x}) \right |$ has a small upper bound of $10^{-7}$. Therefore, the quantization error $\tilde{x}$ with infinite bit precision is almost uniformly distributed.

Consider the difference between $f_{b}(\tilde{x})$ and $f_{\infty}(\tilde{x})$ for $|\tilde{x}|<\frac{1}{2}d$:
\begin{align}
&\left |f_{\infty}(\tilde{x})-f_{b}(\tilde{x})\right|
\nonumber\\&=\frac{1}{\sqrt{2\pi}\sigma}\left(\sum_{k=-{\infty}}^{-2^b}e^{-\frac{(kd+\tilde{x})^2}{2\sigma^2}}+\sum_{k=2^b}^{\infty}e^{-\frac{(kd+\tilde{x})^2}{2\sigma^2}}\right )\nonumber\\&\le\frac{2}{\sqrt{2\pi}\sigma}\sum_{k=2^b}^{\infty}e^{-\frac{(kd-\frac{1}{2}d)^2}{2\sigma^2}}.
\end{align}

Let $A(x)$ denote the function $A(x)=e^{-\frac{x^2}{2\sigma^2}}$. When $d\in \left [\sigma/(2^b-1),\sigma\right]$, $A(x)$ is a convex function for $x \ge (2^b-1)d \ge \sigma$. Therefore, we can find the upper bound of the difference between $f_{b}(\tilde{x})$ and $f_{\infty}(\tilde{x})$
\begin{align}
\left |f_{\infty}(\tilde{x})-f_{b}(\tilde{x})\right|\nonumber&\le \frac2{\sqrt{2\pi}\sigma}\frac1d\int_{(2^b-1)d}^{+\infty}e^{-\frac{x^2}{2\sigma^2}}dx\nonumber\\&=\frac1d\operatorname{erfc}\left(\frac{(2^b-1)d}{\sqrt{2} \sigma}\right).
\end{align}

Thus, the relative discrepancy between $\tilde{x}$ and that of a uniformly distributed variable in PDF is upper bounded as
\begin{align}\label{upperbound}
\frac{\left |\frac{1}{d}- f_b(\tilde{x})\right|}{\frac{1}{d}}&\le d\left|f_{\infty}(\tilde{x})-f_{b}(\tilde{x})\right|+|\delta_{\infty}(\tilde{x})|\nonumber\\&\le\mathrm{erfc}\left(\frac{(2^b-1)d}{\sqrt{2} \sigma}\right)+10^{-7}
\end{align}

\subsection{Proof of Eq. (\ref{lower})}\label{Proofremark}
We can write the LQ decomposition on the matrix $\mathbf{H}$ in the form of the Gram-Schmidt Orthogonalization algorithm as
\begin{equation}
\resizebox{0.92\hsize}{!}{$
\begin{bmatrix}
\mathbf{h}_1
\\\mathbf{h}_2
\\
\vdots
\\\mathbf{h}_K
\end{bmatrix}
=
\begin{bmatrix}
  \sqrt{\mathbf{u}_1\mathbf{u}_1^H} & 0 & \cdots & 0 \\
 \frac{\mathbf{h}_2\mathbf{u}_1^H}{\sqrt{\mathbf{u}_1\mathbf{u}_1^H}} &  \sqrt{\mathbf{u}_2\mathbf{u}_2^H} & \cdots & 0 \\
\vdots & \vdots & \ddots & \vdots \\
  \frac{\mathbf{h}_K\mathbf{u}_1^H}{\sqrt{\mathbf{u}_1\mathbf{u}_1^H}} &   \frac{\mathbf{h}_K\mathbf{u}_2^H}{\sqrt{\mathbf{u}_2\mathbf{u}_2^H}} & \cdots &   \sqrt{\mathbf{u}_K\mathbf{u}_K^H}
\end{bmatrix}
\begin{bmatrix}
\frac{\mathbf{u}_1}{\sqrt{\mathbf{u}_1\mathbf{u}_1^H}}
\\\frac{\mathbf{u}_2}{\sqrt{\mathbf{u}_2\mathbf{u}_2^H}}
\\
\vdots
\\\frac{\mathbf{u}_K}{\sqrt{\mathbf{u}_K\mathbf{u}_K^H}}
\end{bmatrix},
$}
\end{equation}
where $\mathbf{u}_i$ denotes the row vector obtained from $\mathbf{H}$ through orthogonalization as
\begin{equation}
\label{u}
\mathbf{u}_i=\mathbf{h}_i-\sum_{j=1}^{i-1} \frac{\mathbf{h}_i\mathbf{u}_j^H}{\mathbf{u}_j\mathbf{u}_j^H}\mathbf{u}_j.
\end{equation}

Therefore, $l_{kk}^2$ has the following expression
\begin{equation}
\label{l_kk_u}
    l_{kk}^2=\mathbf{u}_k\mathbf{u}_k^H.
\end{equation}

Eq. (\ref{u}) and (\ref{l_kk_u}) demonstrate that $\mathbf{u}_k$ and $l_{kk}^2$ are both solely dependent on the first $k$ row of the channel matrix $\mathbf{H}$. Therefore, when the number of users $K$ increases yet the channel state for the first $k$ users remains fixed, $\mathbf{u}_k$ and $l_{kk}^2$ remain unchanged.

Let $\mathbf{f}_i$ denote the $i$th column vector of matrix $\mathbf{F}_{q}$. As $K$ increases from $K_1$ to $K_2$ $(K_1<K_2 \le N)$, we can prove that
\begin{align}
&\rho_k^{K_1}-\rho_k^{K_2}\nonumber \\
=&\frac{\gamma \bar{l}_{kk}^2}{\gamma\left\|\tilde{\mathbf{h}}_k\mathbf{F}_{q,K_1}\right \|_2^2+\sigma_h^2}-\frac{\gamma \bar{l}_{kk}^2}{\gamma\left\|\tilde{\mathbf{h}}_k\mathbf{F}_{q,K_2}\right \|_2^2+\sigma_h^2}\nonumber \\
=&\frac{ \gamma^2\bar{l}_{kk}^2
\tilde{\mathbf{h}}_k \left ( \sum_{i=K_1+1}^{K_2}\mathbf{f}_i\mathbf{f}_i^H\right )\tilde{\mathbf{h}}_k^H
}{\left( \gamma\left\|\tilde{\mathbf{h}}_k\mathbf{F}_{q,K_1}\right \|_2^2+\sigma_h^2 \right )\left( \gamma\left\|\tilde{\mathbf{h}}_k\mathbf{F}_{q,K_2}\right \|_2^2+\sigma_h^2 \right )} \ge 0.
\end{align}

Therefore, additional users lead to additional interference and decrease the SINR performance for the RRAM-based THP algorithm.

\subsection{Proof of Theorem 1}\label{Proof_pro1}
From the property we prove in Appendix \ref{Proofremark}, the expectation of $\rho$ drops to its lower bound when the number of users is equal to the number of antennas $K=N$, which takes the form
\begin{align}
\mathbb{E}(\rho_k) &\ge \mathbb{E}(\breve{\rho}_k) 
= \mathbb{E}\left (\frac{\gamma\bar{\mathbf{u}}_k\bar{\mathbf{u}}_k^H}{\gamma \tilde{\mathbf{h}}_k\tilde{\mathbf{h}}_k^H+\sigma_h^2} \right ) \nonumber\\
\label{pro1_eq1}
&=\mathbb{E}\left (\frac{\gamma\bar{\mathbf{h}}_k \left ( \mathbf{I}_N -\mathbb{E}\left (\sum_{j=1}^{j<k}\frac{\bar{\mathbf{u}}_j^H\bar{\mathbf{u}}_j}{\bar{\mathbf{u}}_j\bar{\mathbf{u}}_j^H} \right )\right ) \bar{\mathbf{h}}_k^H}{\gamma \tilde{\mathbf{h}}_k\tilde{\mathbf{h}}_k^H+\sigma_h^2} \right ),
\end{align}
where $\bar{\mathbf{u}}_i$ denotes the row vector obtained from $\bar{\mathbf{H}}$ through orthogonalization. Let matrix $\mathbf{M}$ denote $\frac{\bar{\mathbf{u}}_k^H\bar{\mathbf{u}}_k}{\bar{\mathbf{u}}_k\bar{\mathbf{u}}_k^H}$. Its element $m_{ij}$ has the expression
\begin{equation}
m_{ij}=\frac{\bar{u}_{ki}^*\bar{u}_{kj}}{\sum_{l=1}^{N} \left |\bar{u}_{kl} \right|^2}.
\end{equation}

Let $g_{n}(\bar{\mathbf{u}}_1,\bar{\mathbf{u}}_2,\dots ,\bar{\mathbf{u}}_n)$ denote the joint PDF of vectors $\bar{\mathbf{u}}_1,\bar{\mathbf{u}}_2,\dots ,\bar{\mathbf{u}}_n$. By mathematical induction, we can prove the following expression holds for all integers $1\le i\le N $
\begin{equation}
\label{MI}
    \resizebox{0.89\hsize}{!}{$g_{n}(\bar{\mathbf{u}}_1,\bar{\mathbf{u}}_2,\dots ,\bar{\mathbf{u}}_n)=g_{n}(\operatorname{T}_i( \bar{\mathbf{u}}_1),\operatorname{T}_i(\bar{\mathbf{u}}_2),\dots ,\operatorname{T}_i(\bar{\mathbf{u}}_n)),$}
\end{equation}
where $\operatorname{T}_i(\cdot)$ denotes an operator that maps a row vector to a new vector by negating its $i$th element while leaving all other elements unchanged. $\operatorname{T}_i(\bar{\mathbf{u}})$ is equivalent to the following equation
\begin{equation}\label{operation_Ti}
    \operatorname{T}_i( \bar{\mathbf{u}})= \bar{\mathbf{u}}\mathbf{D}_i,
\end{equation}
where $\mathbf{D}_i$ is a diagonal matrix with $-1$ in the $i$th diagonal position and $1$ elsewhere on the diagonal. Evidently, $\mathbf{D}_i$ is a Hermitian unitary matrix.

\ \ Base Case ($n=1$): From Eq. (\ref{u}), we obtain that $\bar{\mathbf{u}}_1=\bar{\mathbf{h}}_1$. Clearly, the property is satisfied.

Inductive Hypothesis: Assume Eq. (\ref{MI}) holds for $n=k$, i.e.,
\begin{equation}
    \resizebox{0.89\hsize}{!}{$g_k(\bar{\mathbf{u}}_1,\bar{\mathbf{u}}_2,\dots ,\bar{\mathbf{u}}_k)=g_k(\operatorname{T}_i( \bar{\mathbf{u}}_1),\operatorname{T}_i(\bar{\mathbf{u}}_2),\dots ,\operatorname{T}_i(\bar{\mathbf{u}}_k)),$}
\end{equation}

Inductive Step ($n=k+1$): Let $\operatorname{O}(\bar{\mathbf{u}}_{1},\bar{\mathbf{u}}_{2},\dots,\bar{\mathbf{u}}_{k},\bar{\mathbf{h}}_{k+1})$ denote the orthogonalization operation in Eq. (\ref{u}) to obtain $\bar{\mathbf{u}}_{k+1}$. 
It has the following equivalent expression
\begin{equation}
    \operatorname{O}( \bar{\mathbf{u}}_1,\bar{\mathbf{u}}_2,\dots,\bar{\mathbf{u}}_k,\bar{\mathbf{h}}_{k+1})=\bar{\mathbf{h}}_{k+1}\operatorname{P}(\bar{\mathbf{u}}_1,\bar{\mathbf{u}}_2,\dots,\bar{\mathbf{u}}_k),\label{equal_u}
\end{equation}
where $\operatorname{P}(\bar{\mathbf{u}}_1,\bar{\mathbf{u}}_2,\dots,\bar{\mathbf{u}}_k)$ is defined as
\begin{equation}
 \operatorname{P}(\bar{\mathbf{u}}_1,\bar{\mathbf{u}}_2,\dots,\bar{\mathbf{u}}_k)=\mathbf{I}_N-\sum_{j=1}^{k} \frac{\bar{\mathbf{u}}_j^H\bar{\mathbf{u}}_j}{\bar{\mathbf{u}}_j\bar{\mathbf{u}}_j^H}.
 \end{equation}
From Eq. (\ref{operation_Ti}), we can derive that
\begin{align}
    \operatorname{P}(\operatorname{T}_i( &\bar{\mathbf{u}}_1),\operatorname{T}_i(\bar{\mathbf{u}}_2),\dots,\operatorname{T}_i(\bar{\mathbf{u}}_k))=\mathbf{I}_N-\mathbf{D}_i^H\sum_{j=1}^{k} \frac{\bar{\mathbf{u}}_j^H\bar{\mathbf{u}}_j}{\bar{\mathbf{u}}_j\bar{\mathbf{u}}_j^H}\mathbf{D}_i\nonumber\\
    &=\mathbf{D}_i^{-1}\operatorname{P}(\bar{\mathbf{u}}_1,\bar{\mathbf{u}}_2,\dots,\bar{\mathbf{u}}_k)\mathbf{D}_i.\label{similarity}
\end{align}

From Eq. (\ref{similarity}), applying the operator $\operatorname{T_i}$ to every input vector, the matrix generated by $\operatorname{P}$ is similar to the original matrix. By Eq. (\ref{equal_u}) and (\ref{similarity}), we can further prove that
\begin{equation}
\resizebox{0.88\hsize}{!}{$
\operatorname{O}(\operatorname{T}_i( \bar{\mathbf{u}}_1),\operatorname{T}_i(\bar{\mathbf{u}}_2),\dots,\operatorname{T}_i(\bar{\mathbf{u}}_k),\operatorname{T}_i(\bar{\mathbf{h}}_{k+1}))=\operatorname{T}_i(\bar{\mathbf{u}}_{k+1})$},
\end{equation}

Let $\varphi _{k+1}( \bar{\mathbf{u}}_1,\bar{\mathbf{u}}_2,\dots ,\bar{\mathbf{u}}_{k},\bar{\mathbf{h}}_{k+1})$ denote the joint PDF of vectors $\bar{\mathbf{u}}_1,\dots ,\bar{\mathbf{u}}_k,\bar{\mathbf{h}}_{k+1}$. We can prove that
\begin{align}
    &g_{k+1}(\operatorname{T}_i( \bar{\mathbf{u}}_1),\operatorname{T}_i(\bar{\mathbf{u}}_2),\dots ,\operatorname{T}_i(\bar{\mathbf{u}}_{k+1}))\nonumber\\
    &=\varphi _{k+1}(\operatorname{T}_i( \bar{\mathbf{u}}_1),\operatorname{T}_i(\bar{\mathbf{u}}_2),\dots ,\operatorname{T}_i(\bar{\mathbf{u}}_{k}),\operatorname{T}_i(\bar{\mathbf{h}}_{k+1}))\left|\operatorname{det}(\mathbf{J}_{\operatorname{T}_i})\right|^{-1}\nonumber\\
    &=\varphi _{k+1}( \bar{\mathbf{u}}_1,\bar{\mathbf{u}}_2,\dots ,\bar{\mathbf{u}}_{k},\bar{\mathbf{h}}_{k+1})\left |\operatorname{det}(\mathbf{J}_{\operatorname{T}_i})\right |^{-1},\label{Induction_half}
\end{align}
where $\mathbf{J}_{\operatorname{T}_i}$ is the corresponding Jacobi matrix defined by the mapping from $\left \{\operatorname{T}_i( \bar{\mathbf{u}}_1),\operatorname{T}_i(\bar{\mathbf{u}}_2),\dots ,\operatorname{T}_i(\bar{\mathbf{u}}_{k}),\operatorname{T}_i(\bar{\mathbf{h}}_{k+1}) \right \}$ to $\left \{\operatorname{T}_i( \bar{\mathbf{u}}_1),\operatorname{T}_i(\bar{\mathbf{u}}_2),\dots ,\operatorname{T}_i(\bar{\mathbf{u}}_{k+1})\right \}$, expressed as
\begin{equation}
\mathbf{J}_{\operatorname{T}_i}= \frac{\partial \operatorname{vec}(\operatorname{T}_i( \bar{\mathbf{u}}_1),\operatorname{T}_i(\bar{\mathbf{u}}_2),\dots,\operatorname{T}_i(\bar{\mathbf{u}}_k),\operatorname{T}_i(\bar{\mathbf{u}}_{k+1}))}{\partial \operatorname{vec}(\operatorname{T}_i( \bar{\mathbf{u}}_1),\operatorname{T}_i(\bar{\mathbf{u}}_2),\dots,\operatorname{T}_i(\bar{\mathbf{u}}_k),\operatorname{T}_i(\bar{\mathbf{h}}_{k+1}))} ,
\end{equation}
where $\operatorname{vec}(\mathbf{x},\mathbf{y},\dots)$ denotes a vectorization operator that concatenates row vectors $\mathbf{x},\mathbf{y},\dots$ into a single column vector by stacking their transposes in sequence.

Let $\mathbf{J}$ denote the Jacobi matrix defined by the mapping from $\{\bar{\mathbf{u}}_1,\bar{\mathbf{u}}_2,\dots ,\bar{\mathbf{u}}_{k},\bar{\mathbf{h}}_{k+1}\}$ to $\{\bar{\mathbf{u}}_1,\bar{\mathbf{u}}_2,\dots ,\bar{\mathbf{u}}_{k},\bar{\mathbf{u}}_{k+1}\}$. We can further derive that
\begin{align}
&\operatorname{det}(\mathbf{J}_{\operatorname{T}_i})\nonumber\\
&=\operatorname{det}\left ( \frac{\partial\operatorname{vec}( \operatorname{T}_i( \bar{\mathbf{u}}_{k+1}))}{\partial \operatorname{vec}(\operatorname{T}_i(\bar{\mathbf{h}}_{k+1}))} \right )\label{Jacobidet_Ti}\\
&=\operatorname{det}\resizebox{0.77\linewidth}{!}{$\displaystyle \left ( \frac{\partial\operatorname{vec}( \operatorname{O}(\operatorname{T}_i( \bar{\mathbf{u}}_1),\operatorname{T}_i(\bar{\mathbf{u}}_2),\dots,\operatorname{T}_i(\bar{\mathbf{u}}_k),\operatorname{T}_i(\bar{\mathbf{h}}_{k+1})))}{\partial \operatorname{vec}(\operatorname{T}_i(\bar{\mathbf{h}}_{k+1}))} \right )$ }\\
&=\operatorname{det}\left ( \frac{\partial\operatorname{vec}( \operatorname{O}( \bar{\mathbf{u}}_1,\bar{\mathbf{u}}_2,\dots,\bar{\mathbf{u}}_k,\bar{\mathbf{h}}_{k+1}))}{\partial \operatorname{vec}(\bar{\mathbf{h}}_{k+1})} \right )\label{delete_Ti}\\
&=\operatorname{det}(\mathbf{J})\label{Jacobidet},
\end{align}
where Eq. (\ref{Jacobidet_Ti}) is derived from the independence of $\bar{\mathbf{h}}_{k+1}$ and $\bar{\mathbf{u}}_{1},\bar{\mathbf{u}}_{2},\dots,\bar{\mathbf{u}}_{k}$, combined with the Jacobian matrix $\mathbf{J}_{\operatorname{T}_i}$ inducing an identity mapping for all variables except the components of $\bar{\mathbf{u}}_{k+1}$. Eq. (\ref{delete_Ti}) is obtained by the similarity derived in Eq. (\ref{similarity}). The derivation of Eq. (\ref{Jacobidet}) resembles that of Eq. (\ref{Jacobidet_Ti}).

Therefore, Eq. (\ref{Induction_half}) is further written as
\begin{align}
    g_{k+1}&(\operatorname{T}_i( \bar{\mathbf{u}}_1),\operatorname{T}_i(\bar{\mathbf{u}}_2),\dots ,\operatorname{T}_i(\bar{\mathbf{u}}_{k+1}))\nonumber\\
    &=\varphi _{k+1}( \bar{\mathbf{u}}_1,\bar{\mathbf{u}}_2,\dots ,\bar{\mathbf{u}}_{k},\bar{\mathbf{h}}_{k+1})\left|\operatorname{det}(\mathbf{J})\right|^{-1}\nonumber\\
    &=g_{k+1}( \bar{\mathbf{u}}_1,\bar{\mathbf{u}}_2,\dots ,\bar{\mathbf{u}}_{k+1}).
\end{align}

Thus, by induction, Eq. (\ref{MI}) holds for all $n\ge1$.

Based on Eq. (\ref{MI}), when $i \neq j$, the expectation of $m_{ij}$ is
\begin{equation}
\mathbb{E}\left (m_{ij} \right )=\mathbb{E}\left (\frac{ u_{ki}^*u_{kj}}{\sum_{l=1}^{N} \left |u_{kl} \right|^2 }\right ) =0.
\end{equation}

When $i=j$, based on equivalence among $\bar{u}_{kl}$, we can obtain that
\begin{equation}
\mathbb{E}\left (m_{ii} \right )
=\frac{1}{N}\mathbb{E}\left (\frac{\sum_{l=1}^{N} \left |\bar{u}_{kl} \right|^2}{\sum_{l=1}^{N} \left |\bar{u}_{kl} \right|^2} \right )=\frac{1}{N}.
\end{equation}

Therefore, Eq. (\ref{pro1_eq1}) can be written as
\begin{align}
\mathbb{E}(\rho_k)&\ge
\frac{N-k+1}{N} \mathbb{E}\left (\frac{\gamma\bar{\mathbf{h}}_k\bar{\mathbf{h}}_k^H}{\gamma \tilde{\mathbf{h}}_k\tilde{\mathbf{h}}_k^H+\sigma_h^2} \right )\\
&\stackrel{(a)}\approx\frac{N-k+1}{N} \frac{\gamma\mathbb{E}\left (\bar{\mathbf{h}}_k\bar{\mathbf{h}}_k^H\right )}{\gamma\mathbb{E}\left ( \tilde{\mathbf{h}}_k\tilde{\mathbf{h}}_k^H\right )+\sigma_h^2} \\
&\stackrel{(b)}\approx\frac{N-k+1}{N} \frac{\gamma\mathbb{E}\left (\mathbf{h}_k\mathbf{h}_k^H\right )}{\gamma\mathbb{E}\left ( \tilde{\mathbf{h}}_k\tilde{\mathbf{h}}_k^H\right )+\sigma_h^2} \\
&\stackrel{(c)}\approx\frac{(N-k+1)\gamma }{\frac{N}{3}\frac{\gamma}{2^{2b-2}} + 1},
\end{align}
where we employ \textbf{Lemma} \textbf{\ref{lemma1}} at point $(a)$, \textbf{Lemma} \textbf{\ref{lemma2}} at point $(c)$. At point $(b)$, we assume that the terms $\tilde{\mathbf{h}}_k\tilde{\mathbf{h}}_k^H$, $\tilde{\mathbf{h}}_k\mathbf{h}_k^H$ is negligible compared to $\mathbf{h}_k\mathbf{h}_k^H$, allowing us to simplify the model without significantly compromising the accuracy.

\subsection{Proof of Theorem 2}\label{Proof_pro2}
We can derive that
\begin{align}
R_k &\ge \mathbb{E}(\log_2(\breve{\rho}_k+1)) \\
&=\mathbb{E}\left(\log_2\left(\frac{\gamma \bar{\mathbf{u}}_k\bar{\mathbf{u}}_k^H}  {\gamma \tilde{\mathbf{h}}_k\tilde{\mathbf{h}}_k^H+\sigma_h^2}+1\right)\right) \\
&\stackrel{(a)}\approx\log_2\left(\frac{\gamma \mathbb{E}(\bar{\mathbf{u}}_k\bar{\mathbf{u}}_k^H)}  {\gamma \mathbb{E}(\tilde{\mathbf{h}}_k\tilde{\mathbf{h}}_k^H)+\sigma_h^2}+1\right) 
\\
&=\log_2\left(\frac{N-k+1}{N}\frac{\gamma \mathbb{E}(\bar{\mathbf{h}}_k\bar{\mathbf{h}}_k^H)}  {\gamma \mathbb{E}(\tilde{\mathbf{h}}_k\tilde{\mathbf{h}}_k^H)+\sigma_h^2}+1\right) 
\\
&\stackrel{(b)}\approx\log_2\left(\frac{N-k+1}{N}\frac{\gamma \mathbb{E}(\mathbf{h}_k\mathbf{h}_k^H)}  {\gamma \mathbb{E}(\tilde{\mathbf{h}}_k\tilde{\mathbf{h}}_k^H)+\sigma_h^2}+1\right) 
\\
&\stackrel{(c)}\approx\log_2\left(\frac{(N-k+1)\gamma}{\frac{N}{3}\frac{\gamma}{2^{2b-2}}}+1\right) ,
\end{align}
where we employ \textbf{Lemma} \textbf{\ref{lemma2}} at point $(c)$. At point $(a)$, we employ the result proposed in \cite{frac14JSTSP}. Similarly to the \textbf{Theorem} \textbf{\ref{Theorem1}}, at point $(b)$, we assume that the terms $\tilde{\mathbf{h}}_k\tilde{\mathbf{h}}_k^H$, $\tilde{\mathbf{h}}_k\mathbf{h}_k^H$ is negligible compared to $\mathbf{h}_k\mathbf{h}_k^H$.

\bibliographystyle{IEEEtran}
\bibliography{references.bib}

\end{document}